\documentclass[sigconf]{acmart}

\AtBeginDocument{%
  }

\setcopyright{acmlicensed}
\copyrightyear{2025}
\acmYear{2025}
\acmDOI{XXXXXXX.XXXXXXX}
\acmConference[SSTD 2025]{18th International Symposium on Spatial and Temporal Databases}{August 18--20, 2025}{Osaka,Japan}

\acmBooktitle{Proceedings of the 18th International Symposium on Spatial and Temporal Databases (SSTD2025), August 18--20, 2025, Osaka, Japan}

\author{Anbang Song}
\affiliation{%
  \institution{Yantai University}
  \country{China}
  \postcode{264005}
}

\author{Ziqiang Yu}
\authornote{Corresponding author, zqyu@ytu.edu.cn}
\affiliation{%
  \institution{Yantai University}
  \country{China}
  \postcode{264005}
}

\author{Wei Liu}
\affiliation{%
  \institution{Yantai University}
  \country{China}
  \postcode{264005}
}

\author{Yating Xu}
\affiliation{%
  \institution{Yantai University}
  \country{China}
  \postcode{264005}
}

\author{Mingjin Tao}
\affiliation{%
  \institution{Yantai University}
  \country{China}
  \postcode{264005}
}

\usepackage{amsmath}
\usepackage{amsfonts}

\usepackage{amssymb}
\usepackage{graphicx}
\usepackage[ruled,linesnumbered]{algorithm2e}
\usepackage{caption}
\usepackage{subfig}
\usepackage{lmodern}
\usepackage{xcolor}
\usepackage{tabularx}

\newcommand{\edited}[1]{\textcolor{black}{#1}}

\newenvironment{editedenv}{\color{black}}{}

\begin{document}
\title{BR$k$NN-light: Batch Processing of Reverse k-Nearest Neighbor Queries for Moving Objects on Road Networks}

\begin{abstract}
The Reverse $k$-Nearest Neighbor (R$k$NN) query over moving objects on road networks seeks to find all moving objects that consider the specified query point as one of their $k$ nearest neighbors. In location based services, many users probably submit R$k$NN queries simultaneously. However, existing methods largely overlook how to efficiently process multiple such queries together, missing opportunities to share redundant computations and thus reduce overall processing costs. To address this, this work is the first to explore batch processing of multiple R$k$NN queries, aiming to minimize total computation by sharing duplicate calculations across queries. To tackle this issue, we propose the BR$k$NN-Light algorithm, which uses rapid verification and pruning strategies based on geometric constraints, along with an optimized range search technique, to speed up the process of identifying the R$k$NNs for each query. Furthermore, it proposes a dynamic distance caching mechanism to enable computation reuse when handling multiple queries, thereby significantly reducing unnecessary computations. Experiments on multiple real-world road networks demonstrate the superiority of the BR$k$NN-Light algorithm on the processing of batch queries.
\end{abstract}



\maketitle

\section{Introduction}

Given a fixed query point and a set of moving objects within a road network, a Reverse $k$-Nearest Neighbor (R$k$NN) query aims to find all moving objects that consider the specified query point as one of the $k$-nearest neighbors. These objects are referred to as R$k$NNs of the query point. R$k$NN queries are critical in location-based services \cite{liu_encyclopedia_2018} and are commonly found in scenarios such as shared mobility \cite{10.1145/3488723}, gaming, facility location, and food delivery. \edited{The problem of optimal facility location, for instance, often relies on kNN-based analysis to select new sites for services \cite{YubaoLIU:152606}.}
For example, a food delivery platform may need to identify which nearby drivers are most likely to accept orders from a specific restaurant. From the driver’s perspective, the restaurant that attracts drivers should be among their $k$ nearest, ensuring the delivery distance remains acceptable for both restaurant and drivers. If we treat the restaurant as the query point and the drivers as moving objects, this task involves querying the R$k$NNs for the query point. Broadly, this problem is a form of spatial task allocation, a core challenge in modern crowd-sourcing systems \cite{7498228, tong2016online}. Since the positions of moving objects constantly change, this work adopts a snapshot-based approach, processing R$k$NN queries based on each object's position at a specific time snapshot, requiring the search algorithm to respond in real time.

\edited{In reality, a location-based service system often needs to handle a massive volume of diverse spatial queries simultaneously \cite{10.14778/3514061.3514064}. Sequentially processing these concurrent queries will lead to substantial redundant computations due to overlapping search areas. Therefore, efficiently processing R$k$NNs in batches has become a critical challenge, especially for queries that exhibit spatial locality or path overlap, a common scenario in urban LBS applications such as shopping districts. Our approach is specifically designed to excel in these situations, effectively optimizing overall execution efficiency for such sets of queries.}

While R$k$NN queries have been extensively researched, most existing approaches focus on optimizing the efficiency of handling individual R$k$NN queries, with little attention given to the efficient batch processing of multiple queries through shared computations. Additionally, even for single-query scenarios, these methods often introduce unnecessary overhead and require further refinement. Specifically, some R$k$NN algorithms \cite{Jin2023,Guohui2010860} designed for general networks face high computational complexity when applied to road networks. To improve this, certain studies construct graph indexes, such as Voronoi diagrams \cite{safar_voronoi-based_2009} or hub label indexes \cite{10.1145/2990192}, or establish local shortcuts between vertices within sub-networks \cite{Hlaing2015,10.1109/HPCC/SmartCity/DSS.2018.00177}, with some recent works even developing sophisticated hierarchical indexes enhanced by shortcuts to handle various spatial queries simultaneously \cite{ZhuoCAO:199610}. However, these methods often fail to fully exploit the unique characteristics of road networks to prune false candidates early. Moreover, in dynamic environments where edge weights (reflecting variable travel times) and moving object positions are frequently updated, maintaining these index structures can incur substantial overhead.

To address these challenges, this paper proposes an algorithm named BR$k$NN-Light, specifically designed for batch processing of R$k$NN queries over moving objects on road networks. Here, the ``Light'' reflects our lightweight approach, which avoids constructing and maintaining complex index structures on the road network. 
BR$k$NN-Light adopts a two-phase ``expansion-verification'' framework.
In the expansion phase, it constructs a shortest path tree starting from each query point in the batch. To avoid traversing the entire network, it further deduces a pruning condition based on the R$k$NN query semantics to filter out invalid expansion branches.
In the verification phase, the algorithm enhances efficiency through the following key points:

1) For each moving object encountered during expansion, it employs a fast verification strategy that leverages an optimized range search based on Euclidean distances between vertices, accelerating the evaluation of potential R$k$NN candidates. In particular, we adopt an R-tree based range search to eliminate the candidate objects that they impossibly become the final R$k$NNs. Unlike traditional methods that typically rely on Minimum Bounding Rectangles (MBRs) for node assessment, we develop R-tree node access and pruning logic based on Minimum Bounding Circles (MBCs). This allows for more precise pruning of the search space.

2) We introduce a cross-phase global distance caching mechanism that maintains a shared distance dictionary $\mathcal{D}$ for all queries in the current batch. This mechanism enables different query tasks to share and reuse computed network distance information, significantly reducing redundant computations and thereby speeding up overall processing, especially when batch query points are spatially clustered or their shortest paths highly overlap.

In summary, the main contributions of this paper are as follows:

\begin{itemize}
    \item \textbf{Batch Processing Framework for R$k$NN Queries:} We are the first to propose and address the problem of batch processing for R$k$NN queries involving moving objects on road networks. We design and implement the BR$k$NN-Light algorithm, a lightweight framework optimized for efficiently handling concurrent R$k$NN queries.

    \item \textbf{Efficient R$k$NN Verification and Pruning Optimization:} For the core verification step in R$k$NN queries, we propose an efficient set of pruning and rapid verification strategies. This approach cleverly adapts and integrates range query techniques in Euclidean space, significantly reducing unnecessary network distance computations and improving the processing efficiency of each query.

    \item \textbf{Computation Sharing Mechanism for Batch Processing:} To further enhance the overall performance of batch processing, we design and implement a dynamic distance caching mechanism. This mechanism effectively captures and reuses distance computation results generated during the processing of different query tasks (especially when query points are spatially close or their paths overlap), thereby substantially reducing redundant computation.

    \item \textbf{Comprehensive Experimental Verification and Performance Evaluation:} We conduct extensive experimental evaluations of the proposed BR$k$NN-Light algorithm on multiple real-world road network datasets. The results show that, compared to existing methods, BR$k$NN-Light demonstrates significant advantages in terms of average response time and total processing time for batch R$k$NN queries, with particularly prominent performance gains on large-scale datasets and in high-concurrency scenarios.
\end{itemize}

The remainder of this paper is organized as follows.
We review related work in Section~\ref{sec:related_works}.
Section~\ref{sec:preliminaries} introduces basic definitions, and Section~\ref{sec:the_problem} formally defines the batch R$k$NN query problem.
Our proposed BR$k$NN-Light algorithm is detailed in Section~\ref{sec:our_rknn}.
Experimental results are presented in Section~\ref{sec:experiments}.
Finally, Section~\ref{sec:conclusion} concludes the paper and outlines future work.

\section{Related Work}
\label{sec:related_works}

This section surveys the research on R$k$NN queries in Euclidean space and on road networks.

\subsection{R$k$NN in Euclidean Space}
Korn and Muthukrishnan first proposed the Reverse Nearest Neighbor (RNN) concept in 2000~\cite{10.1145/335191.335415}, later extending it to R$k$NN. Their RNN-tree used precomputed distance relationships for indexing but required maintaining two indexes in dynamic scenarios. Yang et al.'s RdNN-tree~\cite{Yang2001AI} needed only a single index but still suffered from the latency of precomputation updates.
To overcome this, snapshot-based algorithms were developed. Stanoi et al.'s 60-degree pruning~\cite{Stanoi2000ReverseNN} divided the search space into six symmetrical regions but faced challenges extending to higher dimensions. Tao et al.'s TPL pruning~\cite{10.5555/1316689.1316754} innovatively combined geometric half-space properties with R-tree traversal for efficient high-dimensional data processing. Wu et al.'s FINCH algorithm~\cite{10.14778/1453856.1453970} reduced the 2D search area. Zhang's Nested Approximate Q-neighbor (NAQ) tree~\cite{Zhang2010NAQtree} merged expansion and query steps and introduced pruning rules, enhancing filtering using verification phase information.

Research on R$k$NN queries in Euclidean space also includes real-time aggregation in data streams~\cite{KORN2002814}, continuous monitoring of dynamic objects~\cite{10.1007/978-3-319-19315-1_27}, and techniques like Conic Section Discrimination (CSD)~\cite{Li03102023} to accelerate verification. Furthermore, significant efforts have been dedicated to privacy-preserving R$k$NN queries, addressing challenges in cloud environments, over encrypted data, or for aggregated queries on crowd-sensed data~\cite{10172058, 9910416, 10175584}. Despite significant progress in Euclidean R$k$NN research~\cite{10.14778/2735479.2735492}, these methods fundamentally rely on Euclidean distance. They cannot accurately capture actual connectivity, path characteristics, or dynamic path costs dictated by road network topology, and are therefore difficult to directly apply to road network scenarios.

\subsection{R$k$NN on Road Networks}

R$k$NN query research on road networks has produced a range of techniques, from fundamental algorithms to advanced applications.
Yiu et al.'s Eager algorithm~\cite{10.5555/1128596.1128765}, an early foundational work, first introduced R$k$NN queries to graph settings, using a Dijkstra-like expansion with pruning.
Subsequent work explored precomputation and indexing to enhance efficiency. Safar et al.~\cite{safar_voronoi-based_2009} used Voronoi cells, but storage and query overheads increased significantly with $k$. Hlaing et al.~\cite{Hlaing2015} proposed methods based on Materialized Path Views (SMPV) and Incremental Euclidean Restriction (IER), reducing search space via graph partitioning and Euclidean distance pruning. Efentakis et al.'s ReHub algorithm~\cite{10.1145/2990192} integrated Hub labeling, synergizing offline precomputation with online query processing.

For continuous R$k$NN queries with moving objects and dynamic road networks, proposed solutions include safe region strategies~\cite{ijgi5120247, Cheema2012} to reduce updates, and dual-layer indexes (e.g., DLM-trees~\cite{Guohui2010860}) or trajectory prediction~\cite{10.1109/TKDE.2017.2776268} for efficient continuous query and update handling. To manage time-varying road networks, some studies use grid partitioning and subgraph merging~\cite{10.1109/HPCC/SmartCity/DSS.2018.00177}.

R$k$NN research has also extended to specific graph types and complex applications. Heuristic pruning and algorithms like SWIFT for directed graphs have been proposed~\cite{10.1007/978-3-319-14977-6_10, 10.1109/ICIS.2015.7166606}. Advanced research expanded query objects from individuals to groups~\cite{Allheeib2022} and explored maximizing R$k$NN influence in socio-geographic networks~\cite{Jin2023}.

However, despite progress, existing methods face common limitations in efficiently processing batch R$k$NN queries on dynamic road networks:
First, many methods relying on complex indexes or precomputation incur high maintenance costs as object positions frequently update, hampering real-time performance.
Second, during candidate verification, existing methods (e.g., IER) often fail to fully optimize Euclidean distance-based pruning efficiency or effectively leverage spatial indexes like R-trees to reduce redundant network distance computations.
Finally, and most critically, these methods generally lack cross-query computation sharing mechanisms for batch processing, forcing independent verification and path computation for each query, leading to substantial repetitive work.
This paper addresses these issues, focusing on efficiency bottlenecks in batch processing.

\section{Preliminaries}
\label{sec:preliminaries}

This paper studies query problems based on a road network $G=(V,E,W)$. $V$ is a set of vertices (intersections, facility points, etc.), \edited{where each vertex is associated with two-dimensional coordinates. $E$ is a set of edges and $W$ is a function that assigns a positive weight to each edge}. \edited{The weight of an edge $e$, denoted as $w(e)$, represents the distance of the road segment.} There are two types of entities in the network: a set of fixed facility points $Q_{f} \subseteq V$, and a set of moving objects $\mathcal{M}$ whose locations change dynamically. The position of each moving object $m \in \mathcal{M}$ is described by the edge $e_m=(u,v)$ it is on and its offset $\text{off}_{m}$ from endpoint $u$. Here, if $u$ and $v$ are represented by integers, it is assumed that $u < v$.

The shortest path distance between two points $v_i, v_j \in V$ is denoted as $SD(v_i,v_j)$. A moving object $m$ is located on an edge $e_m=(u,v)$. The shortest path distance from $m$ to a facility point $q \in Q_{f}$ is $SD(m, q) = \min(\text{off}_{m} + SD(u, q), w(e_m)-\text{off}_{m} + SD(v, q))$.

\begin{definition}[$k$-Nearest Neighbor Query]
    Given a moving object $m \in \mathcal{M}$, a set of facilities $Q_{f}$, and an integer $k$, $kNN(m, Q_{f}, k)$ returns the $k$ facilities in $Q_{f}$ that are closest to $m$. The result set $kNN_{res} \subseteq Q_{f}$ satisfies: (1) $|kNN_{res}|=k$; (2) $\forall q' \in kNN_{res}, \forall q'' \in Q_{f} \setminus kNN_{res}, SD(m,q') \le SD(m,q'')$. This is denoted concisely as $kNN(m)$.
\end{definition}
\begin{definition}[Reverse $k$-Nearest Neighbor Query]
    \label{def:rknn_single}
    Given a facility $q_{f} \in Q_{f}$, a set of moving objects $\mathcal{M}$, and an integer $k$, $RkNN(q_{f}, \mathcal{M}, k)$ returns all moving objects in $\mathcal{M}$ that consider $q_{f}$ as one of their $k$-nearest neighbors, i.e., $\{m \in \mathcal{M} \mid q_{f} \in kNN(m, Q_{f}, k)\}$. This is denoted concisely as $RkNN(q_{f})$.
\end{definition}

Table~\ref{tab:notation} summarizes the main symbols used throughout this paper for quick reference.

\begin{table}[htbp]
    \centering
    \caption{Summary of Notation}
    \label{tab:notation}
    \begin{tabularx}{\columnwidth}{@{}lX@{}}
        \toprule
        Symbol & Description \\
        \midrule
        $P$ & Vertex to 2D coordinate mapping, e.g., $P(v)$. \\ 
        $\mathcal{M}$ & Set of moving objects. \\ 
        $m$ & A moving object in $\mathcal{M}$. \\ 
        $e_m=(u,v)$ & Edge of moving object $m$, with $u,v \in V$. \\ 
        $\text{off}_m$ & Offset of $m$ on edge $e_m$ from $u$. \\ 
        $Q_f$ & Set of all facility points (subset of $V$). \\ 
        $q_f$ & A facility vertex in $Q_f$, often a query point. \\ 
        $Q_q$ & Set of query facilities for batch query ($Q_q \subseteq Q_f$). \\ 
        $SD(x,y)$ & Shortest path distance in network $G$. \\ 
        $ED(x, y)$ & Euclidean distance between vertex $x, y$. \\
        $\mathcal{D}$ & Global SSSP result cache. \\
        \bottomrule
    \end{tabularx}
\end{table}

\section{Problem Definition}
\label{sec:the_problem}

Based on the R$k$NN query defined in the previous section, this section formally defines the core problem addressed in this paper: batch Reverse $k$-Nearest Neighbor queries.

\begin{definition}[Batch Reverse $k$-Nearest Neighbor Query]
    \label{def:batch_R$k$NN}
    Given a road network $G=(V,E,W)$, a set of query facilities $Q_{q} \subseteq Q_{f}$ (where $Q_f$ is the set of all facilities), a set of moving objects $\mathcal{M}$, and a positive integer $k$. A batch Reverse $k$-Nearest Neighbor query aims to compute, for each facility $q_{f} \in Q_{q}$, its set of reverse $k$-nearest neighbor moving objects $RkNN(q_{f}, \mathcal{M}, k)$.

    According to the R$k$NN definition (\ref{def:rknn_single}), for a single query facility $q_{f} \in Q_{q}$, its set of reverse $k$-nearest neighbor moving objects is:
    $$RkNN(q_{f}, \mathcal{M}, k) = \{ m \in \mathcal{M} \mid q_{f} \in kNN(m, Q_{f}, k) \}$$
    where $kNN(m, Q_{f}, k)$ are the $k$-nearest neighbors of moving object $m$ from the set of all facilities $Q_{f}$, with distances computed based on $SD(m,q)$.

    The final output of the batch query is a set of $|Q_{q}|$ pairs, each containing a query facility $q_{f} \in Q_{q}$ and its corresponding set of reverse $k$-nearest neighbor moving objects $RkNN(q_{f},\allowbreak \mathcal{M}, k)$. This can be formally represented as:
    $$\text{BRkNN}(Q_{q}, \mathcal{M}, k) = \{ (q_{f}, RkNN(q_{f}, \mathcal{M}, k)) \mid q_{f} \in Q_{q} \} $$
\end{definition}

We use a food delivery scenario as an example to illustrate how batch Reverse $k$-Nearest Neighbor queries can match restaurants with potential delivery personnel. Figure \ref{fig:delivery-map} shows a road network where black line segments represent roads, blue dots represent intersections, red dots represent couriers (moving objects, $o_1$ to $o_4$), and yellow stars represent restaurants (query points, $s_1$ to $s_4$). Their positions are distributed on different road segments.

\begin{figure}[htbp]
    \centering
    \includegraphics[width=0.78\columnwidth]{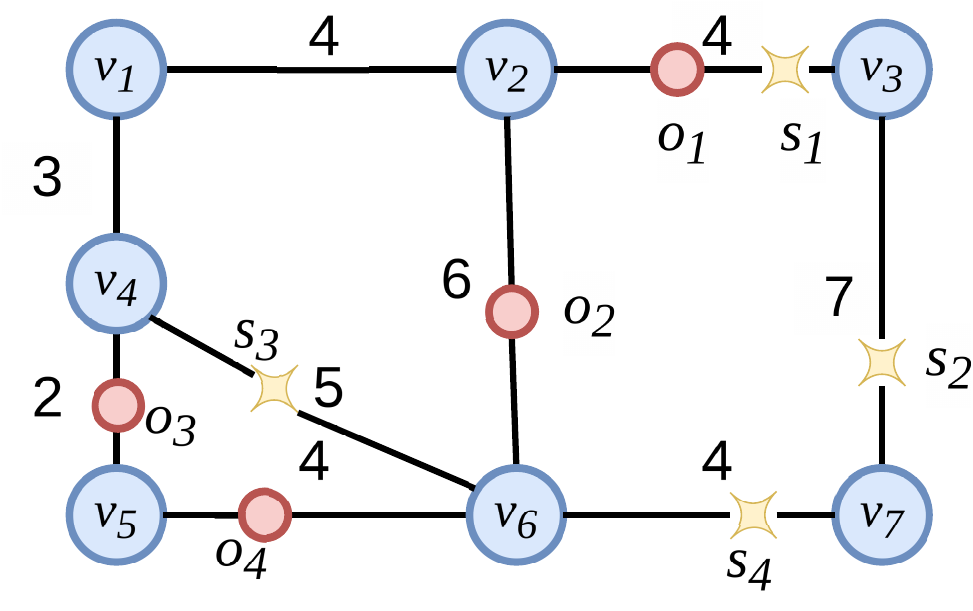}
    \caption{Example of a road network for a food delivery system. Assumptions: Restaurants (yellow stars) are located at vertices of the graph; couriers (red dots) are located on road segments, defined by an edge and an offset.}
    \label{fig:delivery-map}
\end{figure}

Such batch R$k$NN queries are typically initiated by a platform to simultaneously assess, from the perspective of numerous merchants, their attractiveness or service potential to surrounding moving objects (such as couriers). For example, the platform can use this information to optimize the intelligent assignment of orders, ensuring that orders are efficiently undertaken by couriers who consider the corresponding merchant a nearby option, or to identify which merchants might face a shortage of delivery capacity, thereby enabling targeted capacity scheduling or merchant support.

Suppose restaurants $s_1$ and $s_3$ simultaneously wish to obtain information about couriers who can quickly respond to their orders; we assume $k=2$ here.
Restaurant (vertex) $s_1$ needs to find those couriers whose respective lists of $k$ (=2) preferred restaurants include $s_1$. Visually, if courier $o_1$'s two nearest restaurants are $s_1$ and $s_2$, and courier $o_2$'s two nearest restaurants are $s_1$ and $s_4$, then both $o_1$ and $o_2$ are reverse 2-nearest neighbors of $s_1$. Similarly, for facility $s_3$, an R2NN query can find couriers who consider it a preferred restaurant (e.g., assume $o_3$ and $o_4$ in the figure satisfy this condition).

\section{Methodology}
\label{sec:our_rknn}

This section details our BR$k$NN-Light algorithm for batch R$k$NN queries on road networks. We first provide an overview of its ``expansion-verification'' framework (Section~\ref{sec:overview}). Subsequent sections elaborate on the expansion phase with pruning (Section~\ref{sec:expansion}), the R-tree based fast verification mechanism (Section~\ref{sec:knn_verification}), and the cross-query computation sharing via distance caching (Section~\ref{sec:shared_computation}).

\subsection{Overview}
\label{sec:overview}

A naive approach to R$k$NN queries on road networks, verifying each moving object $m \in \mathcal{M}$ via a $kNN(m)$ query, is inefficient. It suffers from redundant network traversals and costly repeated shortest path computations, especially for large datasets and real-time demands.

To address these challenges, BR$k$NN-Light employs an ``Expansion-Verification'' framework.
The expansion phase explores the network from query points, using effective pruning (Section~\ref{sec:expansion}) to drastically reduce the search space.
The verification phase then accurately assesses R$k$NN eligibility for candidates using fast $k$NN membership tests (Section~\ref{sec:knn_verification}).
To systematically enhance efficiency and minimize redundancy across the batch, the framework leverages a shared Single-Source Shortest Path (SSSP) result cache $\mathcal{D}$ (Section~\ref{sec:shared_computation}). These components work synergistically for efficient and accurate batch R$k$NN processing.

\subsection{Expansion}
\label{sec:expansion}

The expansion phase identifies potential R$k$NN candidates by exploring the network from each query facility $q_f$. To avoid exhaustive verification, an effective pruning strategy is crucial.

\subsubsection{Basic Idea of Pruning Strategy}

Many vertices, due to their topological properties, cannot be R$k$NNs of a query point $q$.
Inspired by Yiu et al.'s pruning property for RNN queries~\cite{10.5555/1128596.1128765}, where descendants of an RNN vertex $v$ \edited{(i.e., vertices whose shortest path from $q$ passes through $v$)} are not RNNs of $q$, we extend this for general R$k$NN queries. Our strategy is based on the following lemma:

\begin{lemma}
    \label{thm:rknn_pruning}
    Let $q$ be a query point and $v$ be a vertex. If $v$ is an R$k$NN of $q$, and the shortest path from another vertex $v' \neq v$ to $q$ passes through $v$, then for $v'$ to be an R$k$NN of $q$, $v$ must be one of $v'$'s $k$-nearest neighbors.
\end{lemma}
\begin{proof}
    Assume $v$ is a reverse $k$-nearest neighbor of $q$, which means $q$ is among the $k$-nearest neighbors of $v$. Since the shortest path from $v'$ to $q$ passes through $v$, we have $SD(v', q) = SD(v', v) + SD(v, q)$. As $SD(v, q) > 0$ (when $v \neq q$), it follows that $SD(v', v) < SD(v', q)$. If $v$ is not among the $k$-nearest neighbors of $v'$, then there exist at least $k$ vertices $u_1, \dots, u_k$ such that $SD(v', u_i) < SD(v', v)$ holds for all $i$. Combined with $SD(v', v) < SD(v', q)$, we get $SD(v', u_i) < SD(v', q)$. This implies that these $k$ vertices are closer to $v'$ than $q$, thus $q$ cannot be one of $v'$'s $k$-nearest neighbors, contradicting the assumption that $v'$ is a reverse $k$-nearest neighbor of $q$. Therefore, $v$ must be among the $k$-nearest neighbors of $v'$.
\end{proof}

This lemma enables powerful pruning: during a Dijkstra-like expansion from $q$ building a shortest path tree, if we visit a vertex $v$ that is an R$k$NN of $q$, then for any descendant $v'$ (where $q \leadsto v \leadsto v'$ is a shortest path), if $v$ is not among $v'$'s $k$-nearest neighbors, $v'$ and its descendants can be safely pruned. This significantly reduces vertices requiring verification.

\subsubsection{Priority Queue-based Expansion}
Algorithm~\ref{alg:expansion} (EPP-MO) details this R$k$NN expansion for a single $q_{f}$. It uses a min-priority queue $\mathcal{Q}$ for Dijkstra-like traversal to identify candidate moving objects, incorporating pruning based on Lemma~\ref{thm:rknn_pruning}. The global SSSP cache $\mathcal{D}$ (Section~\ref{sec:shared_computation}) accelerates distance computations and pruning decisions.

\begin{algorithm}[htbp]
    \caption{Expansion Phase for Finding R$k$NN Moving Objects (EPP-MO)}
    \label{alg:expansion}
    \SetAlgoLined
    \KwIn{Query facility $q_{f} \in Q_{q}$, Set of all moving objects $\mathcal{M}$, integer $k$}
    \KwOut{Set of R$k$NN moving objects $R_{mov} \subseteq \mathcal{M}$}
    \SetKwFunction{FExpand}{expand}
    \SetKwProg{Fn}{Function}{}{}
    \Fn{\FExpand{$q_{f}$, $\mathcal{M}$, $k$}}{
        $R_{mov} \gets \emptyset$\;
        $\mathcal{P} \gets \emptyset$\;
        $\mathcal{M}_{processed} \gets \emptyset$\;
        $\mathcal{Q} \gets \text{min-heap with } (0, q_{f})$\; \label{line:init_vars_expansion}
        \While{$\mathcal{Q} \neq \emptyset$}{
            $(d_u, u) \gets \mathcal{Q}.\text{pop}()$\;
            \If{$u \in \mathcal{P}$}{
                \textbf{continue}\;
            }
            $\mathcal{P}.\text{insert}(u)$\;

            \ForEach{edge $e=(u,v')$ incident to $u$}{
                \ForEach{moving object $m \in \mathcal{M}$ such that $m.edge = e$ and $m \notin \mathcal{M}_{processed}$}{
                    \If{VRQ$(q_{f}, k, m)$}{ 
                        $R_{mov}.\text{add}(m)$\;
                    }
                    $\mathcal{M}_{processed}.\text{add}(m)$\;
                }
            }

            bool $\text{canPrunePath} = \textbf{false}$\;

            \If{NOT $\text{canPrunePath}$}{
                Push unvisited adjacent vertices of $u$ to $\mathcal{Q}$ with distances $d_u + w_{uv'}$\; \label{line:push_neighbors_expansion} 
            }
        }
        \Return{$R_{mov}$}\;
    }
\end{algorithm}

EPP-MO expands from $q_{f}$, extracting the closest unprocessed vertex $u$ from $\mathcal{Q}$. Processed vertices (in $\mathcal{P}$) are skipped. Moving objects $m$ on edges incident to $u$ are verified via VRQ($q_{f}, k, m$) (Section~\ref{sec:knn_verification}); R$k$NNs are added to $R_{mov}$. The algorithm then decides whether to expand to $u$'s neighbors based on pruning conditions (Lemma~\ref{thm:rknn_pruning}, aided by cache $\mathcal{D}$). Unpruned, unprocessed neighbors $v'$ are added to $\mathcal{Q}$. This process of ordered expansion, leveraging shared cache $\mathcal{D}$ and pruning, significantly reduces unnecessary traversals until $\mathcal{Q}$ is empty.

\subsubsection{Pruning Example Analysis}

Figure~\ref{fig:prune-example} (left) shows a simple road network with 6 vertices and 5 edges. Here, we consider the query point $q$, points of interest, and intersections all as vertices. To intuitively demonstrate the query process, we explain it using a reverse nearest neighbor query with $q_{f}$ as the query point. We maintain a priority queue $\mathcal{Q}$, initialized with \texttt{<0, $q_{f}$>}.

\begin{figure}[htbp]
    \centering
    \begin{minipage}[t]{0.505\columnwidth}
        \centering
        \subfloat{
            \includegraphics[width=\textwidth]{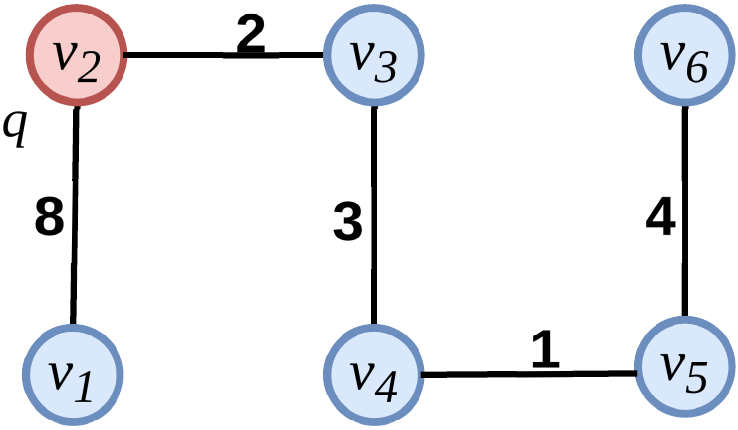}
        }
    \end{minipage}
    \hfill
    \begin{minipage}[t]{0.45\columnwidth}
        \centering
        \subfloat{
            \includegraphics[width=\textwidth]{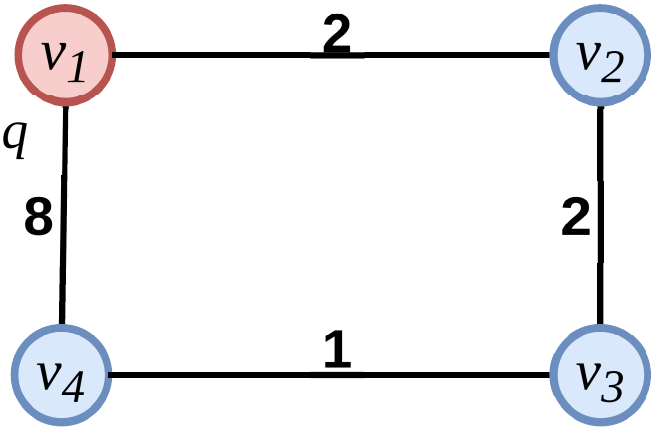}
        }
    \end{minipage}
    \caption{Example of applying pruning strategy during expansion process.}
    \label{fig:prune-example}
\end{figure}

\edited{First, we pop vertex $q_f$ from $\mathcal{Q}$, mark it as processed, and add its adjacent points $v_1$ and $v_3$ to $\mathcal{Q}$. Then, we pop the point of interest $v_3$ and perform a verification process to determine if its nearest neighbor set includes query point $q_f$. Clearly, $q_f$ is one of $v_3$'s nearest neighbors, so we add $v_3$ to the result set and continue expanding through $v_3$. Subsequently, vertex $v_4$ is processed similarly. However, since $v_4$ is not one of $q_f$'s reverse nearest neighbors, we pop it from the queue and do not expand further from it. Finally, we visit the last vertex $v_1$ in $\mathcal{Q}$. Since $VRQ(v_1)$ is $\textbf{true}$, we add it to $\mathcal{R}$, and the query process ends. In this example, we only verified $v_1$, $v_3$, and $v_4$. We pruned the children of $v_4$ from the shortest path tree, so the expansion phase does not visit nodes $v_5$ and $v_6$, avoiding invalid accesses.}

When a vertex is excluded from the shortest path tree starting at $q$, it means it will not be processed again, even if it might appear in other branches of this shortest path tree. Figure~\ref{fig:prune-example} (right) shows a simplified example of this situation. In one step, because $v_3$'s nearest neighbors do not include query point $q$, we mark $v_3$'s adjacent point $v_4$ as processed (or effectively pruned from $v_3$'s branch), thereby excluding $v_4$ from further expansion along this path. In the next iteration, we might pop vertex $v_4$ from $\mathcal{Q}$ if it is directly adjacent to $q$ via another path. But if $v_4$ had already been marked as processed (or pruned as a descendant of a pruned node), it would be skipped. The query process ends. Keep in mind that the purpose of our pruning is to exclude vertices for which the shortest path to query point $q$ passes through an intermediate vertex $v$, and $v$ is not a reverse $k$-neighbor of $q$ (or $v$ is an R$k$NN of $q$ but is not a $k$NN of its descendant). This pruning strategy can significantly reduce unnecessary verification operations and improve algorithm efficiency.

\edited{The expansion process naturally terminates when the priority queue becomes empty. In sparse scenarios where no moving objects are found nearby, the expansion may cover a larger portion of the network. However, our pruning strategies (Lemma~\ref{thm:rknn_pruning}) remain effective in cutting off unpromising paths. For practical applications, an additional termination condition, such as a maximum search radius or a hop count limit, could be easily integrated to prevent excessive exploration in vast, empty regions.}

\subsection{Verification}
\label{sec:knn_verification}

After expansion, the verification phase accurately determines if each candidate moving object $m \in \mathcal{M}$ is an R$k$NN of a query facility $q_{f} \in Q_{q}$, i.e., if $q_{f} \in kNN(m, Q_{f}, k)$. We propose a fast Euclidean distance-based \edited{verification}, accelerated by the SSSP cache $\mathcal{D}$ (Section~\ref{sec:shared_computation}) for any required network $k$NN computations.

\subsubsection{Fast $k$NN Membership Test}
\label{sec:quick_knn_membership_check}

Instead of computing full $k$NN lists, we only need to check if $q_f$ is among $m$'s $k$-nearest neighbors. The following lemma facilitates this using a Euclidean range query with radius $SD(m, q_f)$, counting facilities within this range.

\begin{lemma}[$k$NN Membership \edited{verification}]
    \label{thm:knn_membership}
    Given $m$ (coordinates $P(m)$), $q_{f}$, and $k$. Let $d_{r} = SD(m, q_{f})$ (its computation benefits from $\mathcal{D}$). If the count of facilities $q' \in Q_{f}$ with $ED(P(m), P(q')) \leq d_{r}$ (i.e., within the Euclidean circle centered at $P(m)$ with radius $d_r$) is $\le k$, then $q_{f} \in kNN(m, Q_{f}, k)$.
\end{lemma}
\begin{proof}
    By contradiction. Assume the number of facilities within the Euclidean circle is $\leq k$, but $q_{f} \notin kNN(m)$. Then, there exist at least $k$ network neighbors $q^*_i$ such that $SD(m, q^*_i) \leq d_{r}$. Since $ED \leq SD$, all these $q^*_i$ are within the Euclidean circle, leading to at least $k$ facilities in the circle. If the number of facilities in the circle is $<k$, or if it is $=k$ and $q_{f}$ (which is also in the circle because $ED(m,q_{f})\leq d_{r}$) must be one of the $k$NNs, a contradiction arises. Thus, the original proposition holds.
\end{proof}

Lemma~\ref{thm:knn_membership} provides an efficient sufficient condition using the network distance $d_r$ (from expansion) for a Euclidean range query. Meeting this condition avoids a full network $k$NN search, improving efficiency. Our R-tree based range counting method is detailed in Section~\ref{sec:our_new_query_method}.

Algorithm~\ref{alg:optimized_validate} (VRQ) implements this quick verification. Given the network distance $d_r = SD(m, q_f)$ (obtained during expansion), VRQ first attempts confirmation using Lemma~\ref{thm:knn_membership} via an R-tree Euclidean range query (lines 1-5, details in Section~\ref{sec:our_new_query_method}). If this quick check is inconclusive, it then performs a standard network $kNN(m, Q_f, k)$ query, potentially using cache $\mathcal{D}$ (Section~\ref{sec:shared_computation}), to determine the result (lines 7-8).

\begin{algorithm}[htbp]
    \caption{verification with Range Query (VRQ)}
    \label{alg:optimized_validate}
    \KwIn{Set of all facilities $Q_{f}$, moving object $m$, query facility $q_{f}$, network distance $d_{r}(m, q_{f})$ from $m$ to $q_{f}$, integer $k$}
    \KwOut{Boolean, indicating if $q_{f} \in kNN(m, Q_{f}, k)$}
    
    $d_r \gets d_{r}(m, q_{f})$\;
    $P_m \gets P(m)$\; 
    
    $\mathcal{Q}_{e} \gets \text{RangeQuery}(T_{Rtree}, P_m, d_r)$\; 
    
    \If{$|\mathcal{Q}_{e}| \leq k$}{ 
        \Return \textbf{true}\;
    }
    
    $kNN_{list} \gets kNN(m, Q_{f}, k)$\; 
    \Return ($q_{f} \in kNN_{list}$)\;
\end{algorithm}

\subsubsection{Optimized R-tree Range Counting}
\label{sec:our_new_query_method}

To support the quick $k$NN membership test (Section~\ref{sec:quick_knn_membership_check}), we need to efficiently count facilities within a Euclidean circle. This circle is centered at $m$ (coordinates $P(m)$) and has a radius $d_r=SD(m, q_{f})$. We use an R-tree indexing all facilities $Q_{f}$ by their 2D coordinates. Crucially, each R-tree node pre-stores the total count of facilities within its MBR. For a static facility set, this count is established at build time and requires no subsequent updates, thus incurring no additional maintenance overhead.

To implement this efficient range counting, Algorithm~\ref{alg:rtree_range_count} (RCF) traverses the R-tree (see Figure~\ref{fig:query-process}). It employs precise geometric pruning based on $\text{minDist}$ and $\text{maxDist}$ from the query center to node MBRs, either pruning nodes, directly using pre-stored counts for fully contained nodes, or recursing for partially overlapping ones, significantly boosting efficiency.

\begin{figure}[htbp]
    \centering
    \includegraphics[width=0.78\columnwidth]{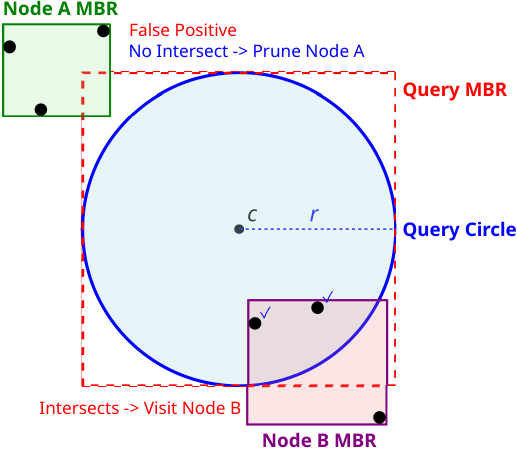}
    \caption{Illustration of an optimized circular range query in an R-tree.}
    \label{fig:query-process}
\end{figure}

\begin{algorithm}[htbp]
    \SetAlgoLined
    \KwIn{Center $P_c$, radius $r$, R-tree $T$}
    \KwOut{Count $N$ of facilities within circle $(P_c, r)$}
    $N \leftarrow 0$\;
    $S \leftarrow [T.\text{root}()]$ \; 
    \While{$S$ is not empty}{
        $curr \leftarrow S.\text{pop}()$\;
        \eIf{$curr$ is a leaf node}{
            \ForEach{facility $q_{obj}$ in $curr$}{
                \If{$ED(P_c, P(q_{obj})) \leq r$}{ 
                    $N \leftarrow N + 1$\;
                }
            }
        }
        { 
            \ForEach{$child$ in $curr.\text{children}()$}{
                \If{$\text{minDist}(P_c, child.\text{MBR}) \leq r$}{ 
                    \If{$\text{maxDist}(P_c, child.\text{MBR}) \leq r$}{ 
                        $N \leftarrow N + child.\text{count}$\; 
                    }
                    \Else{ 
                        $S.\text{push}(child)$\;
                    }
                }
            }
        }
    }
    \Return{$N$}\;
    \caption{Range Count for Facilities (RCF)}
    \label{alg:rtree_range_count}
\end{algorithm}

\subsection{Shared Computation via SSSP Caching}
\label{sec:shared_computation}

Batch R$k$NN processing can still incur redundant Single-Source Shortest Path (SSSP) computations when the same moving object $m$ is verified against multiple query facilities $q_f$. To mitigate this, we introduce a lightweight dynamic SSSP cache $\mathcal{D}$. This map, initialized empty per batch, stores for each SSSP-source moving object $m$ (key) its ``distance map'' $\mathcal{D}_m$ (value), which in turn maps graph vertices $v$ to their SSSP-finalized distances $SD(m,v)$.

When $kNN(m, Q_{f}, k)$ is needed, a cache hit for $\mathcal{D}_m$ directly provides distances. On a miss, a standard SSSP computation from $m$ is performed, and the resulting $\mathcal{D}_m$ is cached for future reuse by any $q_f$. Assuming a static network snapshot per batch, cached distances remain valid. This dynamically built, on-demand cache significantly improves performance, especially when many objects are verified.

\begin{lemma}{Optimality of Cached Distances}
    \label{thm:cache-optimality}
    For any SSSP source $m \in \mathcal{M}$, once its SSSP computation completes, every entry $\mathcal{D}_m[v]$ in the cached distance map $\mathcal{D}_m$ represents the true network shortest path distance $SD(m,v)$.
\end{lemma}
\begin{proof}
    The correctness of this lemma directly stems from the correctness of the SSSP algorithm employed. Standard Dijkstra's algorithm guarantees that when a vertex $v$ is extracted from the priority queue, its computed distance $dist[v]$ is the shortest path length from the source to $v$. Our caching mechanism stores these determined shortest distances after the SSSP algorithm has completed or reached its termination condition. Therefore, the distance values in the cache inherit the optimality of the SSSP algorithm. When a moving object $m$ (acting as an SSSP source) is located on an edge, its distance calculation $SD(m,v)$ to a graph vertex $v$ is achieved by decomposing it into the distance to the edge's endpoints plus the shortest path distance from an endpoint to $v$. This decomposition, also based on the optimal substructure property of shortest paths, does not affect the optimality of the final distance.
\end{proof}

\edited{
    In practice, this cache $\mathcal{D}$ can be efficiently implemented using a key-value data structure, such as a hash map, which was cleared before processing each new batch.
}

\subsection{Time Complexity Analysis}
\label{sec:time_complexity}

The time complexity of BR$k$NN-Light is influenced by its expansion, verification, and SSSP caching mechanisms. Let $|V|$, $|E|$, $|\mathcal{M}|$, $|Q_f|$, and $|Q_q|$ denote the number of vertices, edges, moving objects, total facilities, and batch queries, respectively.

The expansion phase (Algorithm~\ref{alg:expansion}) executes $|Q_q|$ Dijkstra-like searches. Pruning (Theorem~\ref{thm:rknn_pruning}) aims to limit each search to explore $V_{exp_i} \ll |V|$ vertices and $E_{exp_i} \ll |E|$ edges, resulting in a cost of roughly $\sum_{i=1}^{|Q_q|} O((|V_{exp_i}| + |E_{exp_i}|)\log|V_{exp_i}|)$.

The verification phase (Algorithm~\ref{alg:optimized_validate}) is invoked for candidate moving objects. Each verification involves an R-tree range count (Algorithm~\ref{alg:rtree_range_count}), costing $T_{RCF} = O(\log |Q_f|)$ on average. If this quick check (Lemma~\ref{thm:knn_membership}) is inconclusive, a full network $k$NN query is performed.

Crucially, the SSSP cache (Section~\ref{sec:shared_computation}) ensures that an SSSP computation from any moving object $m$, costing $T_{SSSP}$ (e.g., $O((|V'|+|E'|)\log|V'|)$ on a relevant subgraph), occurs at most once per batch. Let $N_{unique\_m} \le |\mathcal{M}|$ be the count of unique objects requiring such an initial SSSP. Subsequent $k$NN verifications for $m$ leverage cached distances, costing significantly less (e.g., $T_{k\_extract\_cache}$).

Thus, the overall time complexity is approximately the sum of total expansion costs, $N_{unique\_m} \times T_{SSSP}$ for initial SSSP computations, and $N_{cand} \times (T_{RCF} + T_{verify\_decision})$ for $N_{cand}$ verification instances. The primary benefit arises from $N_{unique\_m}$ being substantially smaller than the number of SSSPs required by non-sharing approaches. This amortization of SSSP costs, combined with effective pruning and fast R-tree checks, leads to the practical efficiency demonstrated experimentally, rather than a tight worst-case bound that fully captures all heuristic gains.

\section{Experiments}
\label{sec:experiments}

This section presents a comprehensive experimental evaluation of BR$k$NN-Light's performance in processing batch R$k$NN queries on road networks. We investigate its behavior under varying parameters, assess the contribution of its core components through ablation studies, and compare it against existing baselines on multiple real-world datasets. The evaluation focuses on key metrics like query processing time to demonstrate the algorithm's effectiveness and practical utility.

\subsection{Experimental Setup} 

\textbf{Environment and Datasets} 
\edited{All experiments were conducted on a workstation equipped with AMD 7950x 5.8GHz, 32GB RAM and Arch Linux.} Algorithms were in C++ (Clang 19.1.7, O3). Each test ran 10 times; averages are reported.
We used four real-world road networks (Table~\ref{tab:datasets}) from~\cite{9thDIMACS}, with intersections as vertices and road segments as edges (weights represent distances).\edited{Default parameters were used on the NY dataset.}

\begin{table}[htbp]
  \centering
  \caption{Road network datasets.} 
  \label{tab:datasets}
  \begin{tabular}{@{}llrr@{}}
    \toprule
    Dataset & Region & $|V|$ & $|E|$ \\
    \midrule
    NY & New York City & 264,346 & 733,846 \\
    COL & Colorado & 435,666 & 1,057,066 \\
    FLA & Florida & 1,070,376 & 2,712,798 \\
    CTR & Central USA & 14,081,816 & 34,292,496 \\
    \bottomrule
  \end{tabular}
\end{table}

\textbf{Query Load}
Query loads were generated systematically. We studied the impact of batch size $|Q_q|$ and parameter $k$. Default values were used unless specified: batch size 100, $10^5$ moving objects ($\mathcal{M}$), and $k=10$. Query facilities ($Q_q$) were randomly sampled from vertices (uniform spatial distribution). Moving objects ($\mathcal{M}$) were randomly placed on edges. All vertices were potential facilities ($Q_f = V$).

\subsection{Evaluating BR$k$NN-Light}

\subsubsection{Impact of Batch Size}
We evaluate performance with varying batch sizes. Figure~\ref{fig:impact_batch_size_total_time} shows total query time (ms, log scale) of our method on four datasets.
On all datasets, total time increases with batch size. The log y-axis reveals sub-linear growth, indicating good scalability. For instance, on NY, a 100x batch increase (10 to 1000) led to a ~35x time increase (8ms to 280ms); on CTR, a similar increase led to a ~41x time increase (150ms to 6200ms). This sub-linear trend is due to distance caching (Sec.~\ref{sec:shared_computation}), which reuses SSSP results, amortizing costs. Thus, our method is efficient for large batch requests.

\begin{figure}[htbp]
    \centering
    \begin{minipage}[t]{0.48\columnwidth}
        \centering
        \includegraphics[width=\linewidth]{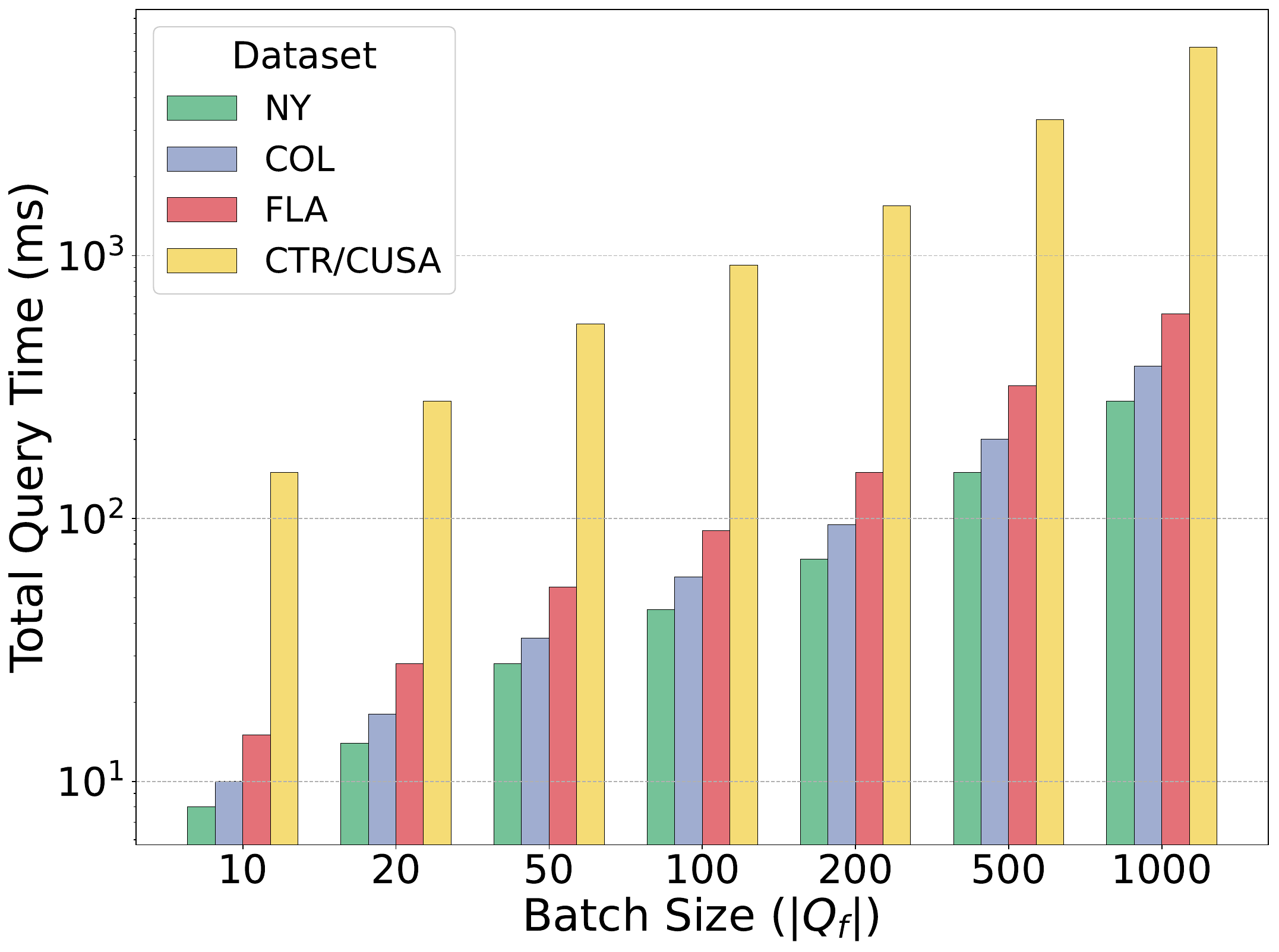}
        \caption{Time vs. batch size.}
        \label{fig:impact_batch_size_total_time}
    \end{minipage}
    \hfill
    \begin{minipage}[t]{0.48\columnwidth}
        \centering
        \includegraphics[width=\linewidth]{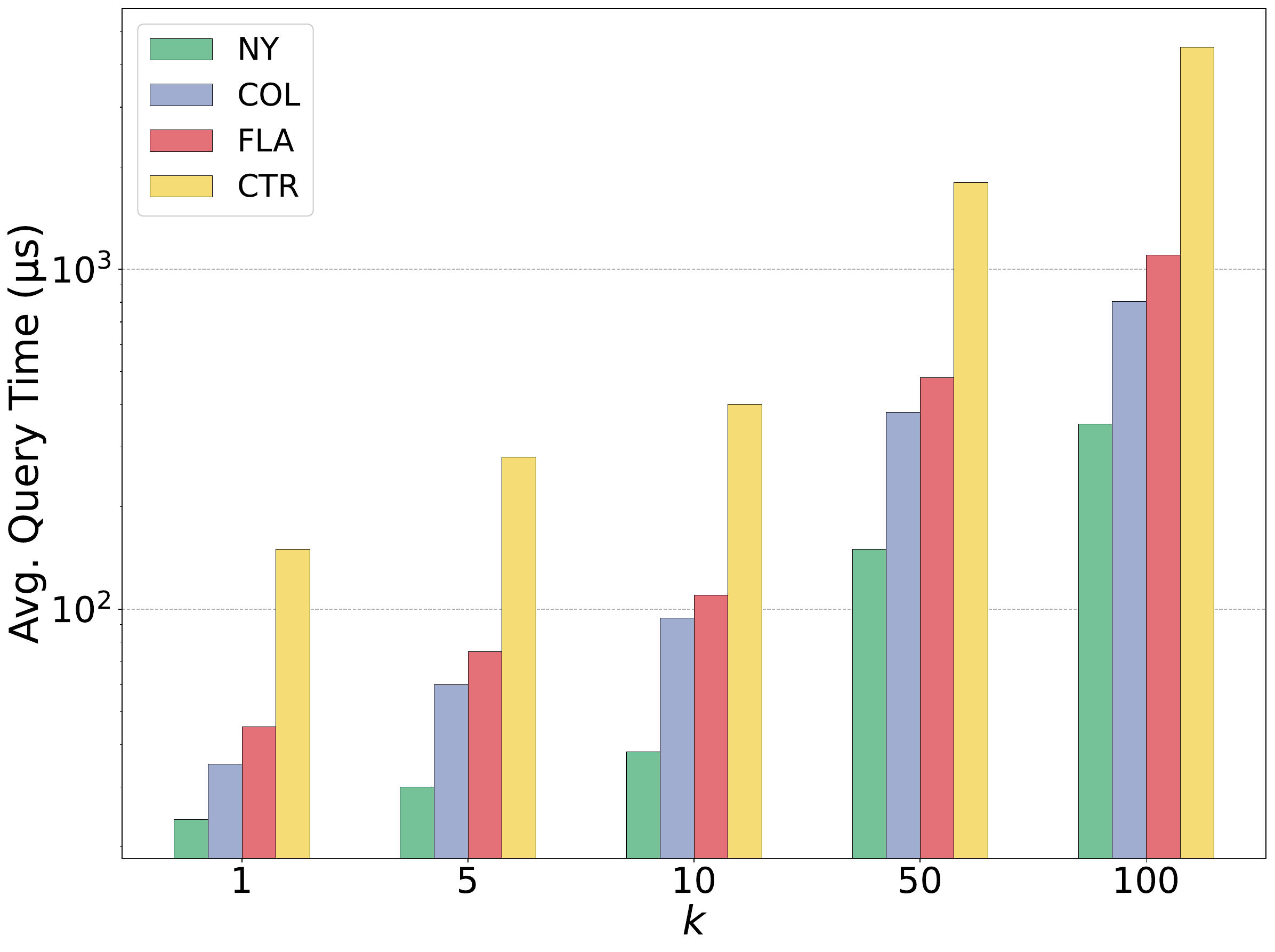}
        \caption{Time vs. $k$.}
        \label{fig:impact_k_value}
    \end{minipage}
\end{figure}

\subsubsection{Impact of $k$ Value}

Figure~\ref{fig:impact_k_value} shows how varying $k$ (1 to 100) affects performance, plotting total query time.
Query time rises with $k$ on all datasets. For small $k$ (1-10), performance degrades mildly. For larger $k$ (50-100), time increases sharply, especially on FLA and CTR. For $k=100$, time increased severalfold over $k=10$. This reflects R$k$NN's inherent complexity as $k$ grows (e.g., larger verification ranges, lower pruning efficiency).

\subsubsection{Impact of Moving Object Count}
\label{sec:exp_impact_object_count}

We examine BR$k$NN-Light's performance as the number of moving objects ($|\mathcal{M}|$) varies from $10^4$ to $10^6$.  Figure~\ref{fig:time_sssp_vs_num_objects} plots total query time and estimated SSSP operations against $|\mathcal{M}|$.

Figure~\ref{fig:time_sssp_vs_num_objects} shows that total query time (bars) increases with $|\mathcal{M}|$. When $|\mathcal{M}|$ grew 100-fold (from $10^4$ to $10^6$), query time rose from 0.025s to 0.180s, a 7.2-fold increase. This sub-linear growth relative to object count indicates good scalability.
The estimated number of SSSP operations (line) also grew sub-linearly, from 500 to 5,500 over the same range of $|\mathcal{M}|$. This controlled increase in SSSP operations, despite a large rise in object numbers, is largely due to effective pruning and the SSSP caching mechanism, which maintained a high utility (e.g., ~45\% hit rate at $|\mathcal{M}|=10^5$).
Our method thus scales well with an increasing density of moving objects. \edited{This sub-linear growth was consistently observed across all four datasets, confirming the robustness of our approach.}

\begin{figure}[htbp]
    \centering
    \begin{minipage}[t]{0.48\columnwidth}
        \centering
        \includegraphics[width=\linewidth]{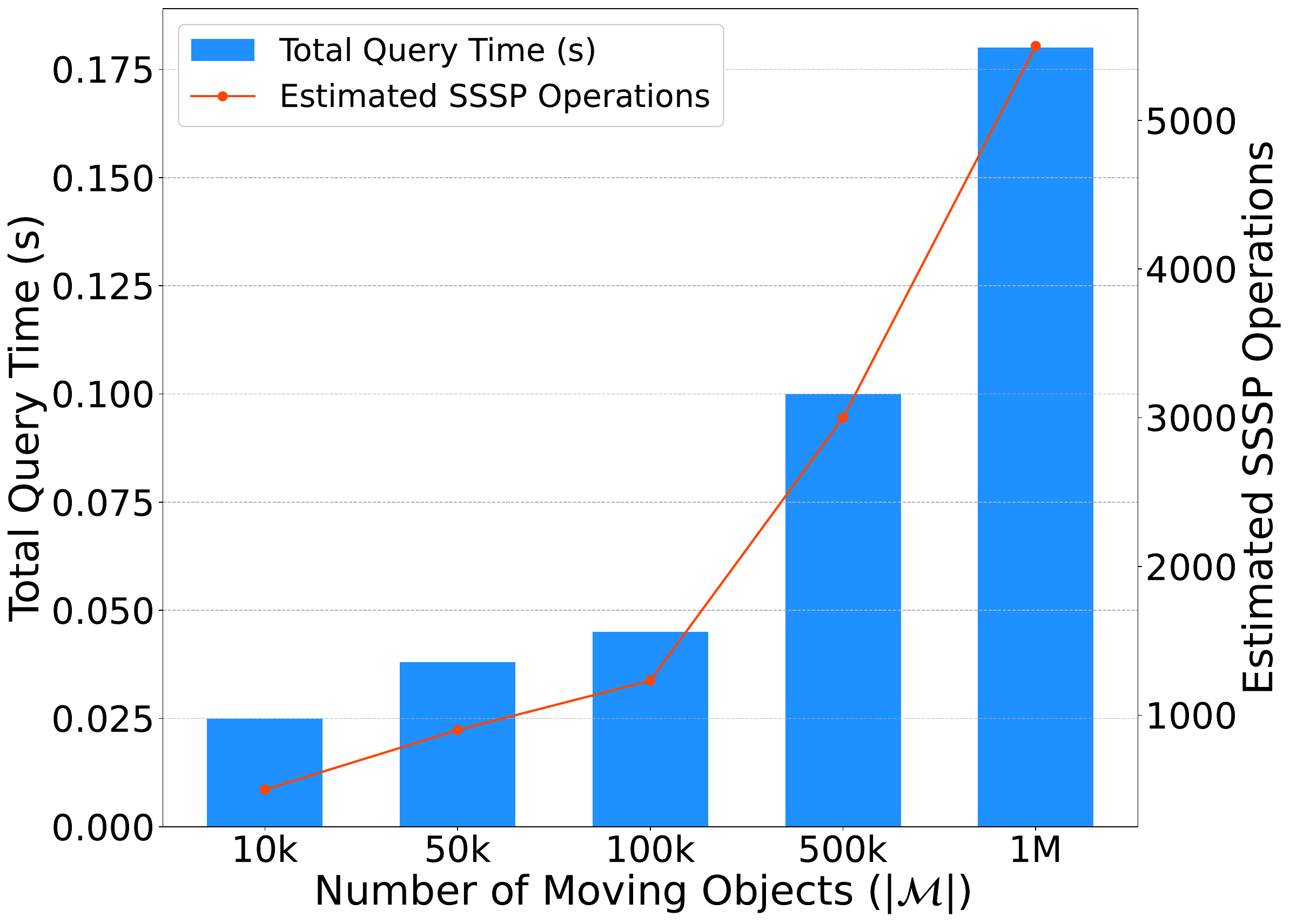}
        \caption{Time \& SSSPs vs. $|\mathcal{M}|$}
        \label{fig:time_sssp_vs_num_objects}
    \end{minipage}
    \hfill
    \begin{minipage}[t]{0.48\columnwidth}
        \centering
        \includegraphics[width=\linewidth]{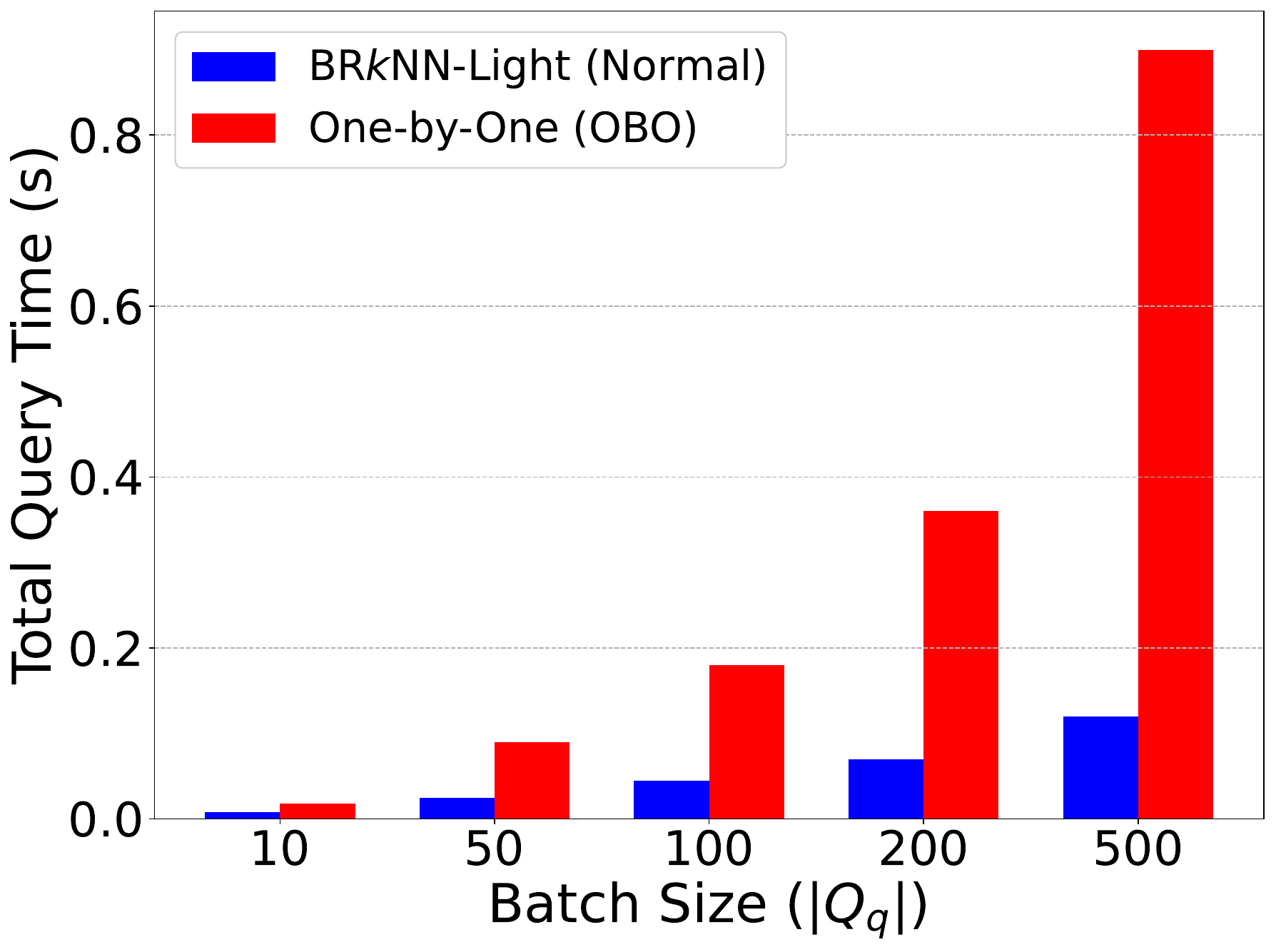}
        \caption{Time: OBO vs. Batch Processing.}
        \label{fig:obo_vs_brknn_comparison_bar}
    \end{minipage}
\end{figure}

\subsubsection{Effectiveness of Batch Processing and Computation Sharing}
\label{sec:exp_batch_effectiveness}

To highlight batching advantages, we compare our full method (BR$k$NN-Light, ``Full'') against a One-by-One (OBO) strategy. OBO processes queries independently, disabling the SSSP cache $\mathcal{D}$ (Sec.~\ref{sec:shared_computation}), but uses identical single-query verifications (QV, MBC R-tree; Sec.~\ref{sec:exp_qv_rtree_optimization}) for a fair comparison focused on sharing benefits. 
Our ``Full'' method significantly outperforms OBO (Figure~\ref{fig:obo_vs_brknn_comparison_bar}), underscoring the inefficiency of processing without computation sharing.

The core of this efficiency lies in SSSP distance caching. Table~\ref{tab:ablation_cache} (NY, batch 100; ``Full'' is BR$k$NN-Light, ``NoCache'' is OBO) shows our ``Full'' method cut query time by 75\% (0.045s vs. OBO's 0.180s) and SSSP computations by 80\% (1,233 vs. 6,171). This is supported by a 45.1\% cache hit rate, meaning nearly half of SSSP results were reused, avoiding redundant traversals. Such reductions in time and SSSPs demonstrate computation sharing's power, with BR$k$NN-Light amortizing costs effectively, especially for larger, overlapping batches.

\begin{table}[htbp]
  \centering
  \caption{Impact of distance caching on NY} 
  \label{tab:ablation_cache}
  \begin{tabular}{@{}lrrr@{}}
    \toprule
    Method & Time (s) & SSSP Ops & Hit (\%) \\ 
    \midrule
    Full & 0.045 & 1,233 & 45.1 \\
    NoCache & 0.180 & 6,171 & N/A \\
    \bottomrule
  \end{tabular}
\end{table}

\subsubsection{Effectiveness of Quick Verification and R-tree Optimization}
\label{sec:exp_qv_rtree_optimization}

Besides batching benefits, our method optimizes single-query verification. We evaluate Quick Verification (QV, Sec.~\ref{sec:quick_knn_membership_check}) and MBC R-tree range counting (Sec.~\ref{sec:our_new_query_method}).
We compare our full method (``Full'') with two variants: (1) "w/o QV": no Quick Verification . (2) ``MBR'': uses MBR R-tree instead of MBC. Experiments used NY dataset (default parameters). Results are in Table~\ref{tab:ablation_quickverify_mbc}.

\begin{table}[htbp]
  \centering
  \caption{QV and R-tree optimization} 
  \label{tab:ablation_quickverify_mbc}
  \begin{tabular}{@{}lrrr@{}}
    \toprule
    Variant & Time (s) & SSSP Ops & R-Nodes \\ 
    \midrule
    Full    & 0.045 & 1,192 & 80 \\
    w/o QV  & 0.120 & 4,511 & N/A \\
    MBR     & 0.049 & 1,534 & 82 \\
    \bottomrule
  \end{tabular}
\end{table}

Table~\ref{tab:ablation_quickverify_mbc} shows QV is critical. Disabling it (w/o QV) increased time from 0.045s (Full) to 0.120s, and SSSP Ops from 1,192 to 4,511. QV avoided ~3,319 SSSP Ops by filtering with Euclidean checks.
Comparing Full (MBC) with MBR-based, MBRs increased time from 0.045s to 0.049s (~8.9\% slower). R-Nodes accessed were similar (80 vs. 82), but SSSP Ops rose from 1,192 (MBC) to 1,534 (MBR), up ~28.7\%. MBR's looser bounds likely caused more false positives, needing more SSSP Ops. MBC's tighter fit reduced SSSP Ops, improving time by ~8.2\% over MBR.
In summary, QV is key for reducing SSSP Ops, and MBC R-tree enhances efficiency with precise pruning.

In summary, the Quick Verification acts as a crucial early filter significantly reducing the number of costly SSSP operations, while the MBC-based R-tree further refines this by enabling more precise geometric pruning than traditional MBR approaches, leading to tangible gains in overall verification efficiency.

\subsection{Comparison with Baselines}

We benchmarked BR$k$NN-Light against four baselines adapted for batch processing: ReHub\cite{10.1145/2990192}, RNNH\cite{Allheeib2022}, and InfZone/SafeZone\cite{LI2024120464}. These methods were executed sequentially for each query in a batch, with any one-time preprocessing costs (like ReHub's) excluded from the online query time measurement.

\begin{editedenv}
    As shown in Figure~\ref{fig:impact_batch_size_time}, our proposed BR$k$NN-Light demonstrates highly competitive performance. It significantly outperforms the baselines on three of the four datasets (COL, FLA, and CTR), showcasing superior scalability, especially on larger and more topologically complex road networks. For instance, on the largest dataset, CTR (Figure~\ref{fig:impact_batch_size_time}(d)), BR$k$NN-Light is approximately 4.17 times faster than its closest competitor, ReHub, for a batch of 500 queries. This highlights the effectiveness of our approach in handling large-scale scenarios.

    A notable exception is the NY dataset, where ReHub's precomputation-based Hub Labeling excels due to the network's highly regular, grid-like structure. This uniformity provides fewer opportunities for our online pruning techniques, giving a narrow edge to precomputation.

    Despite this special case, the overall strong performance of BR$k$NN-Light stems from its core design. The SSSP distance cache effectively amortizes computation costs across batch queries, leading to sub-linear time growth. This, combined with our lightweight verification and pruning, makes BR$k$NN-Light a robust and highly effective solution for batch R$k$NN queries on general-purpose road networks.
\end{editedenv}

\begin{figure}[htbp]
    \centering
    \subfloat[NY]{
        \includegraphics[width=0.47\columnwidth]{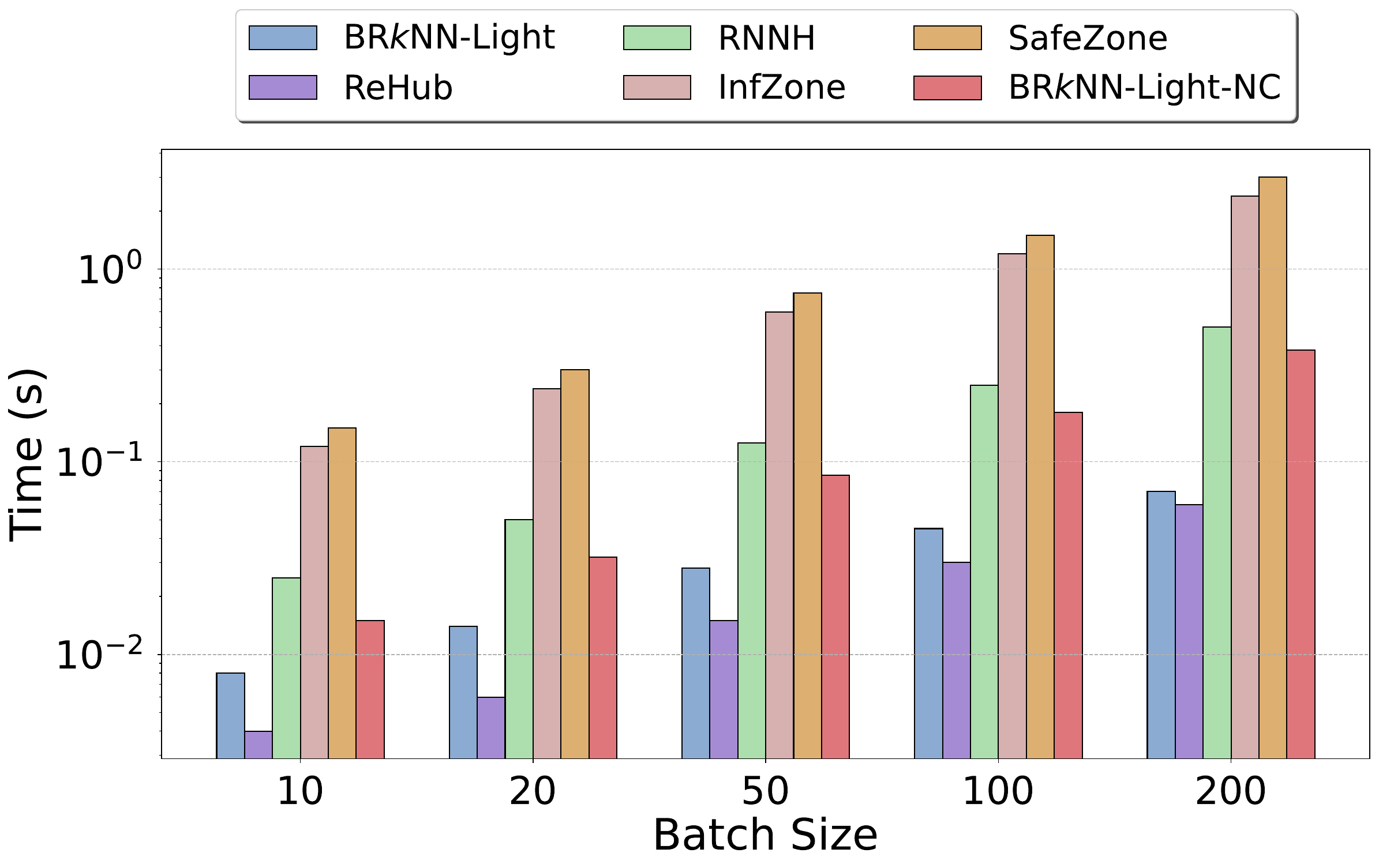}
        \label{fig:ny_time_vs_batchsize} 
    }
    \hfill
    \subfloat[COL]{
        \includegraphics[width=0.47\columnwidth]{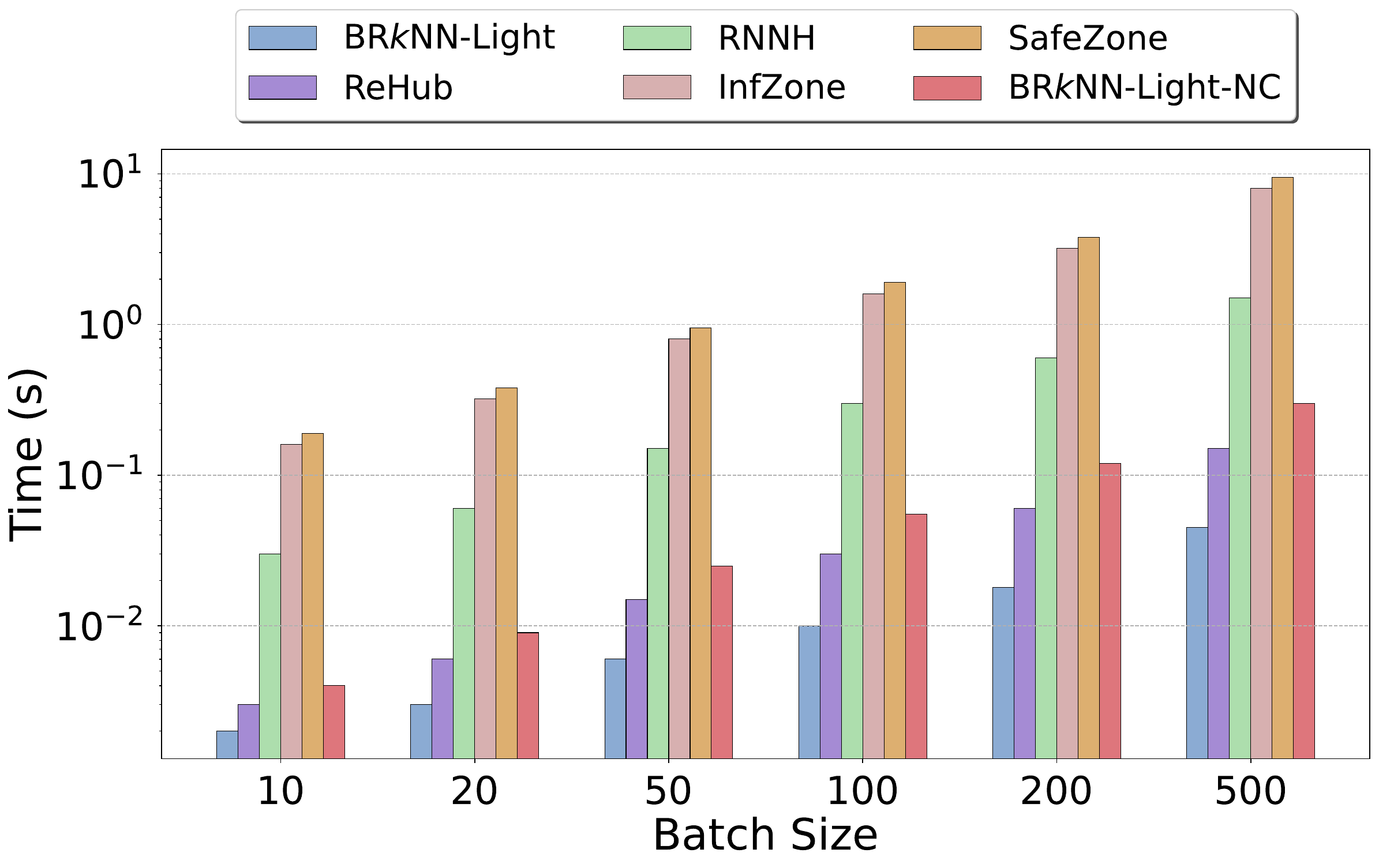}
        \label{fig:col_time_vs_batchsize}
    }
    \vspace{1ex} 
    \subfloat[FLA]{
        \includegraphics[width=0.47\columnwidth]{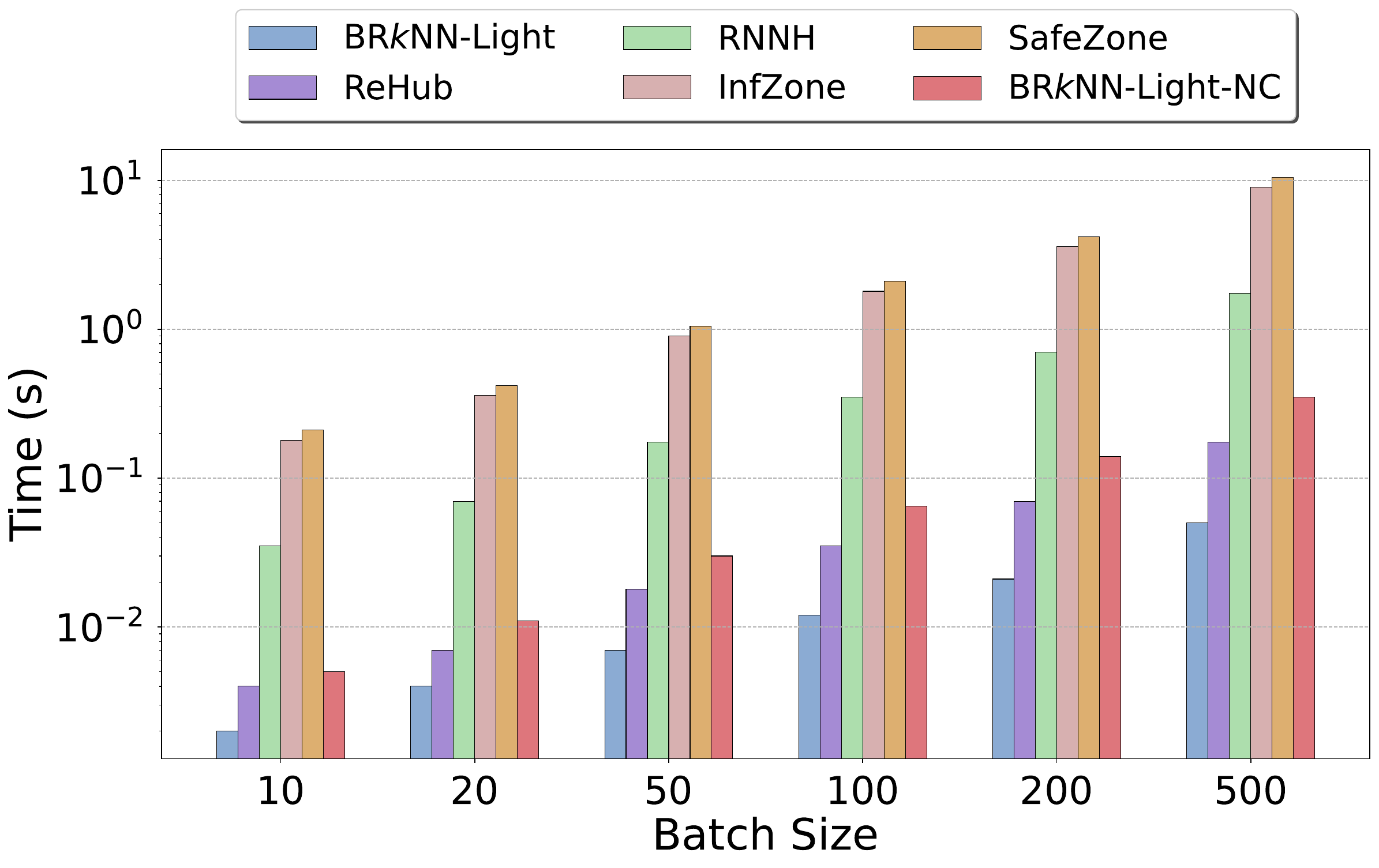}
        \label{fig:fla_time_vs_batchsize}
    }
    \hfill
    \subfloat[CTR]{
        \includegraphics[width=0.47\columnwidth]{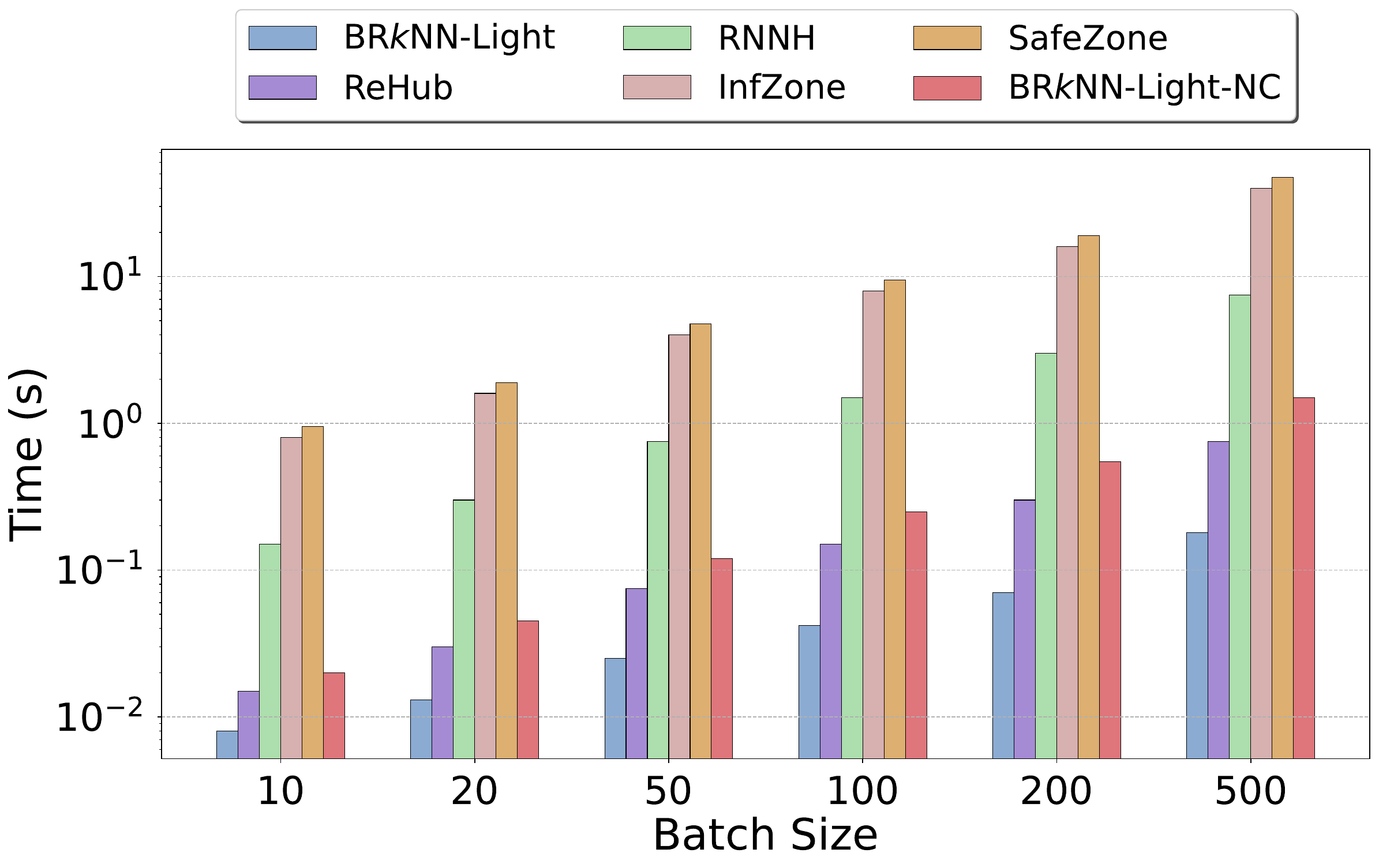}
        \label{fig:ctr_time_vs_batchsize}
    }
    \caption{Total query time vs. batch size (comparison with baselines).} 
    \label{fig:impact_batch_size_time}
\end{figure}

\subsection{Experimental Summary}
\label{sec:exp_summary}
In summary, our extensive experiments demonstrate BR$k$NN-Light's significant advantages. The algorithm scales effectively with increasing batch sizes and moving object counts, primarily due to its SSSP distance caching that enables substantial computation sharing and amortizes costs, leading to sub-linear query time growth. Ablation studies confirmed the critical impact of this caching mechanism and the Quick Verification with MBC R-tree optimization. Consequently, BR$k$NN-Light consistently and substantially outperforms adapted baseline methods across all tested real-world road networks and parameters, proving its suitability for efficiently processing batch R$k$NN queries.

\section{Conclusion}
\label{sec:conclusion}

\edited{This paper tackles the previously unexplored problem of efficiently processing batch R$k$NN queries on road networks. We propose BR$k$NN-Light, a lightweight online framework that strategically shares computation across queries using a dynamic distance cache, avoiding the high maintenance costs of precomputed indexes. Experiments confirm its strong performance on large, complex networks, though we note its limitations on simple grid-like roadmaps where precomputation-based methods are superior. Future work will target adaptive batching and dynamic environments. Our approach provides a practical and efficient solution for modern, large-scale location-based services.}

\section{Acknowledgments}
This work was supported by NSFC Grants [No.62172351].

\section{AI-Generated Content Acknowledgement}
We utilized the Gemini 2.5 Pro model (by Google) to assist in preparing this manuscript. Its application included language polishing, translation, grammar checks, and help with writing and modifying the Python code for our experimental figures. All AI-generated content was carefully reviewed and edited by the authors. The core ideas, analysis, and conclusions are entirely the work of the authors, who bear full responsibility for the paper's accuracy and integrity.

\bibliographystyle{ACM-Reference-Format}
\bibliography{ref}

@book{liu_encyclopedia_2018,
  address   = {New York, NY},
  editor    = {Liu, Ling and Özsu, M. Tamer},
  isbn      = {978-1-4614-8266-6 978-1-4614-8265-9},
  language  = {en},
  publisher = {Springer New York},
  title     = {Encyclopedia of {Database} {Systems}},
  urldate   = {2025-03-23},
  year      = {2018}
}

@article{10.1145/335191.335415,
  address    = {New York, NY, USA},
  author     = {Korn, Flip and Muthukrishnan, S.},
  issn       = {0163-5808},
  issue_date = {June 2000},
  journal    = {SIGMOD Rec.},
  month      = may,
  number     = {2},
  numpages   = {12},
  pages      = {201–212},
  publisher  = {Association for Computing Machinery},
  title      = {Influence sets based on reverse nearest neighbor queries},
  volume     = {29},
  year       = {2000}
}

@inproceedings{10.5555/1316689.1316754,
  author    = {Tao, Yufei and Papadias, Dimitris and Lian, Xiang},
  booktitle = {Proceedings of the Thirtieth International Conference on Very Large Data Bases - Volume 30},
  isbn      = {0120884690},
  location  = {Toronto, Canada},
  numpages  = {12},
  pages     = {744–755},
  publisher = {VLDB Endowment},
  series    = {VLDB '04},
  title     = {Reverse kNN search in arbitrary dimensionality},
  year      = {2004}
}

@inproceedings{Stanoi2000ReverseNN,
  author    = {Ioana Stanoi and Divyakant Agrawal and A. El Abbadi},
  booktitle = {ACM SIGMOD Workshop on Research Issues in Data Mining and Knowledge Discovery},
  title     = {Reverse Nearest Neighbor Queries for Dynamic Databases},
  year      = {2000}
}

@article{10.5555/1128596.1128765,
  address    = {USA},
  author     = {Yiu, Man Lung and Papadias, Dimitris and Mamoulis, Nikos and Tao, Yufei},
  issn       = {1041-4347},
  issue_date = {April 2006},
  journal    = {IEEE Trans. on Knowl. and Data Eng.},
  month      = apr,
  number     = {4},
  numpages   = {14},
  pages      = {540–553},
  publisher  = {IEEE Educational Activities Department},
  title      = {Reverse Nearest Neighbors in Large Graphs},
  volume     = {18},
  year       = {2006}
}

@article{10.1145/2990192,
  address    = {New York, NY, USA},
  articleno  = {1.13},
  author     = {Efentakis, Alexandros and Pfoser, Dieter},
  issn       = {1084-6654},
  issue_date = {2016},
  journal    = {ACM J. Exp. Algorithmics},
  month      = nov,
  numpages   = {35},
  publisher  = {Association for Computing Machinery},
  title      = {ReHub: Extending Hub Labels for Reverse k-Nearest Neighbor Queries on Large-Scale Networks},
  volume     = {21},
  year       = {2016}
}

@inproceedings{Yang2001AI,
  author  = {Yang, Congyun and Lin, King-Ip},
  isbn    = {0-7695-1001-9},
  journal = {Proceedings - International Conference on Data Engineering},
  month   = {02},
  pages   = {485-492},
  title   = {An index structure for efficient reverse nearest neighbor queries},
  year    = {2001}
}

@article{safar_voronoi-based_2009,
  author  = {Safar, Maytham and Ibrahimi, Dariush and Taniar, David},
  issn    = {1432-1882},
  journal = {Multimedia Systems},
  month   = oct,
  number  = {5},
  pages   = {295--308},
  title   = {Voronoi-based reverse nearest neighbor query processing on spatial networks},
  volume  = {15},
  year    = {2009}
}

@article{Guohui2010860,
  author   = {Li Guohui and Li Yanhong and Li Jianjun and LihChyun Shu and Yang Fumin},
  issn     = {0306-4379},
  journal  = {Information Systems},
  number   = {8},
  pages    = {860-883},
  title    = {Continuous reverse k nearest neighbor monitoring on moving objects in road networks},
  volume   = {35},
  year     = {2010}
}

@incollection{KORN2002814,
  address   = {San Francisco},
  author    = {Flip Korn and S. Muthukrishnan and Divesh Srivastava},
  booktitle = {VLDB '02: Proceedings of the 28th International Conference on Very Large Databases},
  editor    = {Philip A. Bernstein and Yannis E. Ioannidis and Raghu Ramakrishnan and Dimitris Papadias},
  isbn      = {978-1-55860-869-6},
  pages     = {814-825},
  publisher = {Morgan Kaufmann},
  title     = {Chapter 70 - Reverse Nearest Neighbor Aggregates Over Data Streams},
  year      = {2002}
}

@article{Allheeib2022,
  author  = {Allheeib, Nasser and Adhinugraha, Kiki and Taniar, David and Islam, Md. Saiful},
  date    = {2022/01/01},
  issn    = {1573-1413},
  journal = {World Wide Web},
  number  = {1},
  pages   = {99--130},
  title   = {Computing reverse nearest neighbourhood on road maps},
  volume  = {25},
  year    = {2022}
}

@article{Li03102023,
  author    = {Yang Li and Mingyuan Bai and Qingfeng Guan and Zi Ming and Xun Liang and Gang Liu and Junbin Gao and},
  journal   = {International Journal of Geographical Information Science},
  number    = {10},
  pages     = {2175--2204},
  publisher = {Taylor \& Francis},
  title     = {CSD-RkNN: reverse k nearest neighbors queries with conic section discriminances},
  volume    = {37},
  year      = {2023}
}

@article{10.1109/TKDE.2017.2776268,
  author  = {Wang, Sheng and Bao, Zhifeng and Culpepper, J. Shane and Sellis, Timos and Cong, Gao},
  journal = {IEEE Transactions on Knowledge and Data Engineering},
  number  = {4},
  pages   = {757-771},
  title   = {Reverse $k$ Nearest Neighbor Search over Trajectories},
  volume  = {30},
  year    = {2018}
}

@inproceedings{10.1007/978-3-319-19315-1_27,
  address   = {Cham},
  author    = {Rahmati, Zahed
               and King, Valerie
               and Whitesides, Sue},
  booktitle = {Combinatorial Algorithms},
  editor    = {Jan, Kratochv{\'i}l
               and Miller, Mirka
               and Froncek, Dalibor},
  isbn      = {978-3-319-19315-1},
  pages     = {307--317},
  publisher = {Springer International Publishing},
  title     = {Kinetic Reverse k-Nearest Neighbor Problem},
  year      = {2015}
}

@article{10.14778/2735479.2735492,
  author     = {Yang, Shiyu and Cheema, Muhammad Aamir and Lin, Xuemin and Wang, Wei},
  issn       = {2150-8097},
  issue_date = {January 2015},
  journal    = {Proc. VLDB Endow.},
  month      = jan,
  number     = {5},
  numpages   = {12},
  pages      = {605–616},
  publisher  = {VLDB Endowment},
  title      = {Reverse k nearest neighbors query processing: experiments and analysis},
  volume     = {8},
  year       = {2015}
}

@article{Jin2023,
  author  = {Jin, Pengfei and Chen, Lu and Gao, Yunjun and Chang, Xueqin and Liu, Zhanyu and Shen, Shu and Jensen, Christian S.},
  issn    = {1573-1413},
  journal = {World Wide Web},
  month   = {July},
  number  = {4},
  pages   = {1567--1598},
  title   = {Maximizing the influence of bichromatic reverse k nearest neighbors in geo-social networks},
  volume  = {26},
  year    = {2023}
}

@inproceedings{10.1109/ICIS.2015.7166606,
  author    = {Qamar, Rizwan and Attique, Muhammad and Chung, Tae-Sun},
  booktitle = {2015 IEEE/ACIS 14th International Conference on Computer and Information Science (ICIS)},
  number    = {},
  pages     = {279-284},
  title     = {A pruning algorithm for reverse nearest neighbors in directed road networks},
  volume    = {},
  year      = {2015}
}

@article{ijgi5120247,
  article-number = {247},
  author         = {Attique, Muhammad and Cho, Hyung-Ju and Jin, Rize and Chung, Tae-Sun},
  issn           = {2220-9964},
  journal        = {ISPRS International Journal of Geo-Information},
  number         = {12},
  title          = {Efficient Processing of Continuous Reverse k Nearest Neighbor on Moving Objects in Road Networks},
  volume         = {5},
  year           = {2016}
}

@inproceedings{10.1109/HPCC/SmartCity/DSS.2018.00177,
  author    = {Li, Jiajia and Li, Yuxian and Shen, Panpan and Xia, Xiufeng and Zong, Chuanyu and Xia, Chenxi},
  booktitle = {2018 IEEE 20th International Conference on High Performance Computing and Communications; IEEE 16th International Conference on Smart City; IEEE 4th International Conference on Data Science and Systems (HPCC/SmartCity/DSS)},
  number    = {},
  pages     = {1064-1069},
  title     = {Reverse k Nearest Neighbor Queries in Time-Dependent Road Networks},
  volume    = {},
  year      = {2018}
}

@inproceedings{10.1007/978-3-319-14977-6_10,
  address   = {Cham},
  author    = {Sahu, Pankaj
               and Agrawal, Prachi
               and Goyal, Vikram
               and Bera, Debajyoti},
  booktitle = {Distributed Computing and Internet Technology},
  editor    = {Natarajan, Raja
               and Barua, Gautam
               and Patra, Manas Ranjan},
  isbn      = {978-3-319-14977-6},
  pages     = {162--173},
  publisher = {Springer International Publishing},
  title     = {Finding RkNN Set in Directed Graphs},
  year      = {2015}
}

@article{Hlaing2015,
  author  = {Aye Thida Hlaing and Tin Nilar Win and Htoo Htoo and Yutaka Ohsawa},
  journal = {Journal of Information Processing},
  number  = {2},
  pages   = {163-170},
  title   = {RkNN Query on Road Network Distances},
  volume  = {23},
  year    = {2015}
}

@phdthesis{Zhang2010NAQtree,
  address      = {Ottawa},
  author       = {Zhang, Ming},
  howpublished = {2 microfiches},
  isbn         = {9780494641415, 049464141X},
  language     = {English},
  note         = {Includes bibliographical references.},
  oclc         = {775853991},
  publisher    = {Library and Archives Canada = Bibliothèque et Archives Canada},
  school       = {University of Calgary},
  title        = {{NAQ-tree}: Effective Index Structure for Similarity Search in High Dimensional Space},
  year         = {2010}
}

@article{10.14778/1453856.1453970,
  author     = {Wu, Wei and Yang, Fei and Chan, Chee-Yong and Tan, Kian-Lee},
  issn       = {2150-8097},
  issue_date = {August 2008},
  journal    = {Proc. VLDB Endow.},
  month      = aug,
  number     = {1},
  numpages   = {12},
  pages      = {1056–1067},
  publisher  = {VLDB Endowment},
  title      = {FINCH: evaluating reverse k-Nearest-Neighbor queries on location data},
  volume     = {1},
  year       = {2008}
}

@article{Cheema2012,
  author  = {Cheema, Muhammad Aamir and Zhang, Wenjie and Lin, Xuemin and Zhang, Ying and Li, Xuefei},
  date    = {2012-02-01},
  issn    = {0949-877X},
  journal = {The VLDB Journal},
  number  = {1},
  pages   = {69--95},
  title   = {Continuous reverse k nearest neighbors queries in Euclidean space and in spatial networks},
  volume  = {21},
  year    = {2012}
}

@article{LI2024120464,
  author   = {Xinyu Li and Arif Hidayat and David Taniar and Muhammad Aamir Cheema},
  issn     = {0020-0255},
  journal  = {Information Sciences},
  pages    = {120464},
  title    = {Continuous monitoring of reverse approximate nearest neighbour queries on road network},
  volume   = {667},
  year     = {2024}
}

@online{9thDIMACS,
  author       = {{DIMACS}},
  note         = {Competition on shortest path algorithms},
  organization = {DIMACS},
  title        = {9th {DIMACS} Implementation Challenge - Shortest Paths},
  url          = {https://www.diag.uniroma1.it//challenge9/competition.shtml},
  urldate      = {2025-05-18},
  year         = {2006}
}

@article{10172058,
  author  = {Zheng, Yandong and Zhu, Hui and Lu, Rongxing and Guan, Yunguo and Zhang, Songnian and Wang, Fengwei and Shao, Jun and Li, Hui},
  journal = {IEEE Transactions on Dependable and Secure Computing},
  number  = {4},
  pages   = {1831-1844},
  title   = {PHRkNN: Efficient and Privacy-Preserving Reverse kNN Query Over High-Dimensional Data in Cloud},
  volume  = {21},
  year    = {2024}
}

@article{9910416,
  author  = {Zheng, Yandong and Lu, Rongxing and Zhang, Songnian and Guan, Yunguo and Wang, Fengwei and Shao, Jun and Zhu, Hui},
  doi     = {10.1109/TDSC.2022.3211870},
  journal = {IEEE Transactions on Dependable and Secure Computing},
  number  = {5},
  pages   = {4387-4402},
  title   = {PRkNN: Efficient and Privacy-Preserving Reverse kNN Query Over Encrypted Data},
  volume  = {20},
  year    = {2023}
}

@article{10175584,
  author  = {Zheng, Yandong and Zhu, Hui and Lu, Rongxing and Guan, Yunguo and Zhang, Songnian and Wang, Fengwei and Shao, Jun and Li, Hui},
  journal = {IEEE Transactions on Information Forensics and Security},
  number  = {},
  pages   = {4285-4299},
  title   = {Efficient and Privacy-Preserving Aggregated Reverse kNN Query Over Crowd-Sensed Data},
  volume  = {18},
  year    = {2023}
}

@article{10.1145/3488723,
  address    = {New York, NY, USA},
  articleno  = {2},
  author     = {Tong, Yongxin and Zeng, Yuxiang and Zhou, Zimu and Chen, Lei and Xu, Ke},
  issn       = {0362-5915},
  issue_date = {March 2022},
  journal    = {ACM Trans. Database Syst.},
  month      = may,
  number     = {1},
  numpages   = {48},
  publisher  = {Association for Computing Machinery},
  title      = {Unified Route Planning for Shared Mobility: An Insertion-based Framework},
  volume     = {47},
  year       = {2022}
}

@article{10.14778/3514061.3514064,
  author     = {Tong, Yongxin and Pan, Xuchen and Zeng, Yuxiang and Shi, Yexuan and Xue, Chunbo and Zhou, Zimu and Zhang, Xiaofei and Chen, Lei and Xu, Yi and Xu, Ke and Lv, Weifeng},
  issn       = {2150-8097},
  issue_date = {February 2022},
  journal    = {Proc. VLDB Endow.},
  month      = feb,
  number     = {6},
  numpages   = {14},
  pages      = {1159–1172},
  publisher  = {VLDB Endowment},
  title      = {Hu-Fu: efficient and secure spatial queries over data federation},
  volume     = {15},
  year       = {2022}
}

@inproceedings{7498228,
  author    = {Tong, Yongxin and She, Jieying and Ding, Bolin and Wang, Libin and Chen, Lei},
  booktitle = {2016 IEEE 32nd International Conference on Data Engineering (ICDE)},
  number    = {},
  pages     = {49-60},
  title     = {Online mobile Micro-Task Allocation in spatial crowdsourcing},
  volume    = {},
  year      = {2016}
}

@article{ZhuoCAO:199610,
  author    = {Zhuo Cao and Chun Cao and Jianqiu Xu and Jingwei Xu and Zhefei Chen and Zi Chen and Xiaoxing Ma},
  journal   = {Frontiers of Computer Science},
  number    = {9},
  publisher = {Front. Comput. Sci.},
  title     = {SCG-tree: shortcut enhanced graph hierarchy tree for efficient spatial queries on massive road networks},
  volume    = {19},
  year      = {2025}
}

@article{YubaoLIU:152606,
  author    = {Yubao Liu and Zitong Chen and AdaWai-Chee Fu and Raymond Chi-Wing Wong and Genan Dai},
  journal   = {Frontiers of Computer Science},
  number    = {2},
  publisher = {Front. Comput. Sci.},
  title     = {Optimal location query based on k nearest neighbours},
  volume    = {15},
  year      = {2021}
}

@inproceedings{tong2016online,
  author    = {Yongxin Tong and Jieying She and Bolin Ding and Libin Wang and Lei Chen},
  title     = {Online Mobile Micro-Task Allocation in Spatial Crowdsourcing},
  booktitle = {Proceedings of the 32nd International Conference on Data Engineering (ICDE)},
  pages     = {49--60},
  year      = {2016},
  address   = {Helsinki, Finland},
  month     = {May}
}

\end{document}